\documentclass{amsart}

\usepackage{amssymb}
\usepackage{enumerate}
\usepackage{graphicx}
\usepackage{amsmath,mathrsfs}
\usepackage{multicol}
\usepackage{tikz}
\usetikzlibrary{automata,positioning}
\usepackage{wrapfig}
\usepackage{bussproofs}
\usepackage{times}
\usepackage{color}
\usepackage[all]{xy}
\usepackage{imakeidx}

 \newtheorem{thm}{Theorem}[section]
 \newtheorem{cor}[thm]{Corollary}
 \newtheorem{lem}[thm]{Lemma}
 \newtheorem{prop}[thm]{Proposition}
 \theoremstyle{definition}
 \newtheorem{defn}[thm]{Definition}
 \theoremstyle{remark}
 
 \newtheorem*{ex}{Example}
 \numberwithin{equation}{section}

\newcommand{\ra}{\Rightarrow}

\newcommand{\naraba}{\rightarrow}

\newcommand{\setof}[1]{\{{#1}\}}

\newcommand{\hkn}{\mathsf{H}\mathbf{(K_n) }}
\newcommand{\gkn}{\mathsf{G}\mathbf{(K_n) }}

\newcommand{\hkdn}{\mathsf{H}\mathbf{(KD_n) }}
\newcommand{\gkdn}{\mathsf{G}\mathbf{(KD_n) }}

\newcommand{\hktn}{\mathsf{H}\mathbf{(KT_n) }}
\newcommand{\gktn}{\mathsf{G}\mathbf{(KT_n) }}

\newcommand{\gksecn}{\mathsf{G}\mathbf{(K^2_n) }}
\newcommand{\gktsecn}{\mathsf{G}\mathbf{(KT^2_n) }}
\newcommand{\gkdsecn}{\mathsf{G}\mathbf{(KD^2_n) }}

\newcommand{\gktnplus}{\mathsf{G}\mathbf{(KT_n^+) }}

\begin{document}

%
%
%
%
%
%
%
%
%

\title{A proof-theoretic approach to uniform
interpolation property of multi-agent modal logic
}
 
\author{Youan Su}

\address{%
School of philosophy\\
Liaoning University\\
No.58, Daoyinan Street, Shenbeixinqu \\
Shenyang, China}

\email{su.youan@lnu.edu.cn}



\date{Sep 15, 2025}

\begin{abstract}
Uniform interpolation property (UIP) is a strengthening of Craig interpolation property. It was first established by  Pitts\cite{pitts1992interpretation}
based on a pure proof-theoretic method. 
UIP in multi-modal  $\mathbf{K_n}$, $\mathbf{KD_n}$ and $\mathbf{KT_n}$ logic have been established by semantic approaches, however, a proof-theoretic approach is still lacking.
B{\'\i}lkov{\'a} \cite{bilkova2007uniform} develops the method in Pitts\cite{pitts1992interpretation} to show UIP in classical modal logic $\mathbf{K}$ and $\mathbf{KT}$.  
This paper further extends  B{\'\i}lkov{\'a}  \cite{bilkova2007uniform}'s systems to  establish the UIP in 
multi-agent modal logic $\mathbf{K_n}$, $\mathbf{KD_n}$ and $\mathbf{KT_n}$. A purely syntactic algorithm is presented to determine a uniform interpolant formula.
It is also shown that quantification over propositional variables can be modeled by UIP in  these systems.
Furthermore, a direct argument to establish UIP without using   second-order quantifiers is also presented.

\end{abstract}

\maketitle

\section{introduction}

We say that a logic $\mathbf{L}$ satisfies the Craig interpolation property if, whenever $A \rightarrow B$ is derivable, there exists an interpolant formula $C$ such that both $A \rightarrow C$ and $C \rightarrow B$ are derivable, and all propositional variables appearing in $C$ are shared by $A$ and $B$.

The {\sl Uniform Interpolation Property (UIP) } is a strengthening of Craig interpolation. 
A logic $\mathbf{L}$ satisfies the  (UIP), if there exist post-interpolant and pre-interpolant formulas satisfying the following conditions.

First, for any formula $A$ and any  propositional variables
    $q_1,\cdots,q_m$ ($m\in\mathbb{N}$, omitted as $\overrightarrow{q_m} $), there exists a formula (called {\sl post-interpolant}) $\mathcal{I}_{post}(A,\overrightarrow{q_m} )$ which is dependently constructed from $A$ and $\overrightarrow{q_m}$, such that:

\begin{enumerate}
    \item $A\rightarrow\mathcal{I}_{post} (A ,\overrightarrow{q_m})  $ is derivable;
    \item  for any formula $B$, if $A\rightarrow B$ is derivable and $\mathsf{V}(A)\cap\mathsf{V}(B) \subseteq\{\overrightarrow{q_m} \}$, then   $\mathcal{I}_{post}(A ,\overrightarrow{q_m})\rightarrow B $ is derivable.
\end{enumerate}

Furthermore, for any formula $B$ and  any  propositional variables
    $r_1,\cdots,r_n$ ($n\in\mathbb{N}$, omitted as $\overrightarrow{r_n} $), there exists a formula (called {\sl pre-interpolant}) $\mathcal{I}_{pre}(B,\overrightarrow{r_n})$ which is dependently constructed from $B$ and $\overrightarrow{r_n}$, such that: 

\begin{enumerate}
    \item $\mathcal{I}_{pre}(B,\overrightarrow{r_n}) \rightarrow B$ is derivable;
    \item  for any formula $A$,  if $A\rightarrow B$ is derivable and $\mathsf{V}(A)\cap\mathsf{V}(B) \subseteq\{\overrightarrow{r_n} \}$ then   $A\rightarrow \mathcal{I}_{pre}(B,\overrightarrow{r_n} ) $ is derivable.
\end{enumerate}

Pitts \cite{pitts1992interpretation} established the UIP as a strengthening of Craig interpolation for intuitionistic propositional logic, based on a sequent calculus that absorbs all structural rules.

In recent years, UIP has been widely studied. For example, it has been shown to correspond to the notion of “forgetting” in knowledge representation and reasoning \cite{Lin1994,hans2009}.

UIP in modal logic $\mathbf{K}$ was shown by Visser \cite{visser1996bisimulations} using bounded bisimulations and by Ghilardi \cite{ghilardi1995algebraic} using an algebraic approach.
Wolter \cite{wolter1997} proved that modal logic $\mathbf{S5}$ has the UIP.
It is also known that $\mathbf{K4}$ and $\mathbf{S4}$ do not satisfy UIP \cite{ghilardi1995undefinability,bilkova2007uniform}.

Regarding multi-agent modal logic, Wolter \cite{wolter1997} showed that UIP for any normal mono-modal logic can be generalized to its multi-agent case.
Fang et al. \cite{FANG201951} proved that $\mathbf{K_n}$
$\mathbf{D_n}$,
$\mathbf{T_n,}$
$\mathbf{K45_n,}$
$\mathbf{KD45_n,}$
$\mathbf{S5_n}$and these systems with common knowledge satisfy UIP.
Alassaf et al. \cite{alassaf_uip_2022} demonstrated UIP for $\mathbf{K_n,D_n,T_n}$ using a resolution-based approach.

There have also been studies of UIP in modal logic based on proof-theoretic approaches.
B{\'\i}lkov{\'a} \cite{bilkova2007uniform} developed the method of Pitts \cite{pitts1992interpretation} to show UIP in classical modal logics $\mathbf{K}$ and $\mathbf{KT}$.
UIP for $\mathbf{KD}$ was shown by Iemhoff \cite{Iemhoff2019_uip}.
UIP for $\mathbf{K}$, $\mathbf{D}$, $\mathbf{T}$, and $\mathbf{S5}$ via nested sequents and hypersequents has been established by van der Giessen et al. \cite{Giessen2024_uip_proof}.
(However, this is not a purely proof-theoretic approach, since semantic notions are used to define uniform interpolation.)

As far as I know, UIP in multi-agent modal logic has not been studied using purely proof-theoretic methods.
However, when Pitts \cite{pitts1992interpretation} established UIP for the first time, he provided a purely proof-theoretic method without using semantic notions.


Craig suggested that his results (namely, Craig interpolation) ``relate suggestive but nonelementary model-theoretic concepts to more elementary proof-theoretic concepts, thereby opening up model-theoretic problems to proof-theoretic methods of attack” \cite{Craig1957}.
In the case of Craig interpolation, Beth’s definability theorem and Robinson’s consistency theorem are related to cut-elimination (cf. Chang and Keisler \cite[chp 2.2]{chang1990model}; Maehara \cite{Maehara1961_interpolation}; Troelstra and Schwichtenberg \cite[p116]{troelstra2000basic}).
In the context of UIP, bisimulation in the semantic approach \cite{visser1996bisimulations} corresponds to terminating proof-search trees in the proof-theoretic approach \cite{pitts1992interpretation,bilkova2007uniform}.
Moreover, the proof-theoretic approach provides a direct method to construct an interpolant formula, which is comparatively difficult to achieve in the bisimulation approach.


This paper extends the single-modal systems $\mathbf{K}$ and $\mathbf{KT}$ studied by B{\'\i}lkov{\'a} \cite{bilkova2007uniform} to multi-modal systems $\mathbf{K_n}$, $\mathbf{KD_n}$, and $\mathbf{KT_n}$.
It provides a proof-theoretic proof of UIP for multi-agent modal logics $\mathbf{K_n}$, $\mathbf{KD_n}$, and $\mathbf{KT_n}$, and presents a purely syntactic algorithm for determining uniform interpolant formulas.
It also shows that quantification over propositional variables can be modeled by UIP in these systems.

Furthermore, in both  B{\'\i}lkov{\'a}\cite{bilkova2007uniform} and Pitts \cite{pitts1992interpretation}, UIP is established via a translation from second-order propositional calculi to first-order propositional calculi.
In this paper, we present a direct argument establishing UIP without using second-order quantifiers.

This paper is organized as follows:

\begin{itemize}
    \item In section \ref{sec:syntax}, we present our syntax. And, we also explain why $\bot$ is needed to be primitive based on Ono\cite{Ono1998};
    \item  Section \ref{sec:sequent cal} introduces the main sequent calculi $\gkn$, $\gkdn$, and $\gktn$, and proves their proof-theoretic properties; Section \ref{sec:main thm of gkn and gkdn} proves the main theorem for $\gkn$ and $\gkdn$, showing UIP for these systems without propositional quantifiers;

    \item Section \ref{sec:sequent cal of T} presents a sequent calculus $\gktnplus$ with a loop-preventing mechanism based on B{\'\i}lkov{\'a} \cite{bilkova2007uniform}, and examines its proof-theoretic properties; Section \ref{sec:main thm of T} proves UIP for $\gktn$;
    
    \item  Section \ref{sec:2nd order} shows that quantification over propositional variables can be modeled by UIP in these systems and provides a translation from second-order to first-order propositional calculi.
    
\end{itemize}

\section{Syntax}
\label{sec:syntax}

We fix a finite set $\mathsf{Agt}$ of agents, a countable set $\mathsf{Prop}$ of propositional variables.

The set of formulas of the language $\mathcal{L}^1$ is defined inductively as:
\begin{center}
$A ::= p   \,| \bot \,| A\wedge A \,|  A\lor A\,| A\rightarrow A| \neg A \,| \Box i A, $\end{center}
where $p\in\mathsf{Prop} $ and $i\in \mathsf{Agt}.$

 Furthermore, the propositional second order modal  language $\mathcal{L}^2$ 
is defined inductively as:
\begin{center}
$A ::= p   \,| \bot \,| A\wedge A \,| A\lor A\,| A\rightarrow A| \neg A \,| \Box i A\, |\forall p A$,\end{center}
where $p\in\mathsf{Prop} $ and $i\in \mathsf{Agt}.$ Greek alphabet in uppercase letters, for example, $\Gamma,\Delta$, will be used to represent multi-sets of formulas. In some cases, Greek alphabet in lowercase letters, for example, $\gamma,\delta$ will also be used to represent formulas.

 In $\mathcal{L}^1$ and $\mathcal{L}^2$, a diamond formula $\Diamond_i A$ is defined as $\neg \Box_i \neg A$. Also, $\top$ is defined as $\neg \bot$. 
 In $\mathcal{L}^2$, $\exists p A$ is defined as $\neg \forall p \neg A$.
 In B{\'\i}lkov{\'a}\cite{bilkova2007uniform},
propositional constants ( $\bot$ or $\top$) are not primitive in the syntax. As was mentioned in Ono \cite{Ono1998}, $\Box \bot\rightarrow \bot$ is not provable in ${\bf K}$, then it  brings in a trouble if we want to eliminate all propositional constants. 
As a result, we put $\bot$ primitive in our syntax.
 
 A formula in the form of $\Box_i A$ is called an outmost $i$-boxed formula. 
Given a finite multiset $\Gamma$ of formulas,  for an arbitrary modal symbol $\Box_i$, $\Gamma^{\natural_i}= \{ \Box_i A |  \Box_i A\in \Gamma \}.$ 
$\Gamma^{\flat_i} = \{    A |  \Box_i A\in \Gamma  \} $.
$\Box_i \Gamma= \{\Box_i A | A\in \Gamma \}$

A substitution of a propositional variable $ p$ with a formula $B$ in a formula $A$ is denoted by $A[p/B]$. In a multiset $\Gamma$ of formulas, $\Gamma[p/B]=\{A[p/B]| A \in \Gamma\}$.

We use $\mathsf{V}(A)$ to denote the set of all propositional variables in a formula $A$. Similarly, given a multiset $\Gamma$ of formulas, $\mathsf{V}(\Gamma) = \{\mathsf{V}(A) | A\in \Gamma\}$.



\begin{defn}
    The {\sl weight} of a $\mathcal{L}^1$-formula $A$, noted as $\mathsf{wt}(A)$ is inductively defined as: 
    \begin{center}
         $\mathsf{wt}(p)=\mathsf{wt}(\bot)=1$
         
    $\mathsf{wt}(\neg A)=\mathsf{wt}(\Box_i A)=\mathsf{wt}(A)+1$
    
     $\mathsf{wt}(A\wedge B)=\mathsf{wt}(A\lor B)= \mathsf{wt}(A\rightarrow B)=\mathsf{wt}(A)+\mathsf{wt}(B)+1$
    \end{center}

    Given a multiset $\Gamma$ of formula, $\mathsf{wt}(\Gamma)$ denotes the sum of all $\mathsf{wt}(A)$ for $A\in \Gamma$.
   
\end{defn}



    

\begin{table}[htb]
\caption{Hilbert system $\hkn$, $\hkdn$ and $\hktn$}
\label{table:hilbert}

\begin{center}
\begin{tabular}{ l l } 
\hline
\multicolumn{2} {c} { Hilbert system $\hkn$ } \\

\multicolumn{2} {c} { All classical propositional tautologies} \\






(K$_n$) & $\Box_i (A \naraba B) \rightarrow (\Box_i A \rightarrow \Box_i B)$      \\


(Nec) & From $A$, infer $\Box_i A$. \\
 

\hline

\multicolumn{2} {c} { Hilbert system $\hkdn$ } \\
\multicolumn{2} {c} { Add the following rule into $\hkn$} \\

(D$_n$) & $\neg \Box_i \bot$ \\

\hline
\multicolumn{2} {c} { Hilbert system $\hktn$ } \\
\multicolumn{2} {c} { Add the following rule into $\hkn$} \\

(T$_n$) & $ \Box_i A \rightarrow A$ \\

\hline

\end{tabular}
\end{center}
\end{table}
\noindent

\begin{defn}
     The Hilbert systems of $\hkn$, $\hkdn$ and $\hktn$ is defined in Table \ref{table:hilbert}.
     Let $\mathbf{L}\in \{\mathbf{K}_n,\mathbf{KD}_n,\mathbf{KT}_n\}$. Given a set $\Gamma \cup \setof{A}$ of formulas, when we write $\Gamma \vdash_{\mathbf{L}} A$,
we  mean that $A$ is derivable from $\Gamma$ in $\mathbf{L}$  (if the underlying Hilbert system is clear from the context, we simply write $\Gamma \vdash A$). In particular, when $\Gamma$ is empty, we simply write $\vdash A$ instead of $\emptyset \vdash A$. 
\end{defn}

\section{UIP in Logic $\mathbf{ K_n}$ and $\mathbf{ KD_n}$}

\subsection{Proof-theoretic properties of Sequent calculi}
 \label{sec:sequent cal}

Next, let us move to Gentzen system.

A  {\it sequent}, denoted by $\Gamma \ra \Delta$, is a pair of finite multisets of formulas. 
The multiset $\Gamma$ is the {\em antecedent} of $\Gamma \ra \Delta$,  
while $\Delta$ is the {\em succedent} of the sequent $\Gamma \ra \Delta$. 
A sequent $\Gamma\ra A$  can be read as ``if all formulas in $\Gamma$ hold then some formulas in $\Delta$ hold.''


The logical rules in the following sequent calculi are the same with those in a system named as {\bf G3cp} \cite[p.49]{Negri2001}.
Modal {\bf K}-rules is an  expansion of single modal rules from   \cite{bilkova2007uniform} (also inspired by \cite{hakli2012does}).

\begin{table}[htb]
\caption{Sequent Calculi $\gkn$, $\gkdn$ and $\gktn$.}
\label{table:gkn}
\begin{center}
\hrule
\begin{tabular}{ll} 
\multicolumn{2}{c}{ Sequent Calculus $\gkn$: } \\

{ \bf Initial Sequents}   & $\Gamma ,p \ra p, \Delta$\hspace{15pt} $\bot, \Gamma\ra \Delta$  \\ 
 \\







{\bf Logical Rules} & 

\AxiomC{$\Gamma \ra \Delta ,A_1$}
\AxiomC{$\Gamma \ra \Delta ,A_2$}
\RightLabel{\scriptsize $(R\wedge )$}
\BinaryInfC{$\Gamma \ra  \Delta ,A_1\wedge A_2$}
\DisplayProof

\AxiomC{$A_1, A_2,\Gamma \ra \Delta$}
\RightLabel{\scriptsize $(L\wedge )$}
\UnaryInfC{$A_1\wedge A_2,\Gamma \ra \Delta$}
\DisplayProof
\\

\,  &

\AxiomC{$\Gamma \ra A_1, A_2$}
\RightLabel{\scriptsize $(R\lor )$}
\UnaryInfC{$\Gamma \ra A_1\lor A_2$}
\DisplayProof

\AxiomC{$A_1,\Gamma \ra \Delta$}
\AxiomC{$A_2,\Gamma \ra \Delta$}
\RightLabel{\scriptsize $(L\lor )$}
\BinaryInfC{$A_1\lor A_2, \Gamma \ra \Delta$}
\DisplayProof

\\

\,  &

\AxiomC{$A_1,\Gamma \ra \Delta, A_2$}
\RightLabel{\scriptsize $(R\rightarrow )$}
\UnaryInfC{$\Gamma \ra \Delta,A_1\rightarrow A_2$}
\DisplayProof

\AxiomC{$\Gamma \ra \Delta, A_1$}
\AxiomC{$A_2,\Gamma \ra \Delta$}
\RightLabel{\scriptsize $(L\rightarrow )$}
\BinaryInfC{$A_1\rightarrow A_2, \Gamma \ra \Delta$}
\DisplayProof

\\

\,  &
\AxiomC{$A, \Gamma \ra  \Delta$}
\RightLabel{\scriptsize $(R\neg )$}
\UnaryInfC{$\Gamma \ra \Delta,\neg A$}
\DisplayProof

\AxiomC{$\Gamma \ra \Delta, A$}
\RightLabel{\scriptsize $(L\neg )$}
\UnaryInfC{$\neg A, \Gamma \ra  \Delta$}
\DisplayProof
\\

{\bf Modal Rule}  &

\AxiomC{$ \Gamma \ra A $ }
\RightLabel{\scriptsize $(\Box_{Kn})$$\dagger$}
\UnaryInfC{$\Sigma, \Box i \Gamma   \ra \Box_i A, \Omega$}
\DisplayProof
\\


\multicolumn{2}{l}{ {\footnotesize $\dagger$: $\Sigma$ contains only propositional variables, $\bot$ or outmost-boxed formulas except $\Box_i$.  $\Omega$ contains only  }
   } \\
   
\multicolumn{2}{l}{ {\footnotesize  propositional variables, $\bot$ or outmost-boxed formulas. }
 }\\

   
   \\
\hline

\multicolumn{2}{c}{ Sequent Calculus $\gkdn$
 } \\
 
\multicolumn{2}{c}{ Adding the following rules  to $\gkn$ 
 } \\

{\bf Modal Rule}  &

\AxiomC{$ \Gamma \ra $ }
\RightLabel{\scriptsize $(\Box_{Dn})$$\ddagger$}
\UnaryInfC{$\Sigma, \Box i \Gamma   \ra \Omega$}
\DisplayProof
\\


\multicolumn{2}{l}{ {\footnotesize $\ddagger$: $\Sigma$ contains only propositional variables, $\bot$ or outmost-boxed formulas except $\Box_i$.  $\Omega$ contains only  }
   } \\
\multicolumn{2}{l}{ {\footnotesize  propositional variables, $\bot$ or outmost-boxed formulas. Also, $\Gamma \neq \emptyset$.}
 }\\

\hline
\multicolumn{2}{c}{ Sequent Calculus $\gktn$
 } \\
\multicolumn{2}{c}{ Adding the following rules  to $\gkn$ 
 } \\

{\bf Modal Rule}  &

\AxiomC{$ \Box_i A, A, \Gamma \ra \Delta $ }
\RightLabel{\scriptsize $(\Box_{Tn})$}
\UnaryInfC{$\Box i A,\Gamma   \ra \Delta$}
\DisplayProof

\\

\hline

\end{tabular}
\end{center}
\end{table}

\begin{defn} 
\label{def:g1_context_principal}
Let $\mathbf{L}\in \{ \mathbf{K_n} , \mathbf{KD_n}, \mathbf{KT_n}\}$ and $ \mathsf{G}(\mathbf{L})$ be one of
the systems of Table \ref{table:gkn}. 

\end{defn}

\begin{defn}
    In $\gkn$, $\gkdn$ and $\gktn$, we say that the formulas (or multisets) not in $\Gamma$ and $\Delta$ are  {\it principal} in all rules except $(\Box_{K_n})$ and $(\Box_{Dn})$ . In the rule of  $(\Box_{K_n})$, the formulas (or multisets) not in $\Sigma,\Omega$ are principal. In the rule of $(\Box_{Dn})$, the formulas in $\Gamma$ are define as principal formulas. 
    A formula (or multiset) is called {\it context} in a rule if it is not principal. 
\end{defn}

\begin{defn}
\label{def:g1_deri}

Let $\mathbf{L}\in \{ \mathbf{K_n} , \mathbf{KD_n}, \mathbf{KT_n} \}$.
A derivation ${\mathcal D}$ in $ \mathsf{G}(\mathbf{L})$ is a finite tree generated by the rules of $ \mathsf{G}(\mathbf{L})$ from the initial sequents of $ \mathsf{G}(\mathbf{L})$. We say that the {\it end sequent} of ${\mathcal D}$ is the sequent in the root node of ${\mathcal D}$. The {\it height} $n$ of a derivation is the maximum length of the branches in the derivation from the end sequent to an initial sequent.
A sequent $\Gamma \ra \Delta$ is {\it derivable} in $ \mathsf{G}(\mathbf{L})$ (notation: $ \mathsf{G}(\mathbf{L}) \vdash \Gamma \ra \Delta$) if it has a derivation ${\mathcal D}$ in  $ \mathsf{G}(\mathbf{L})$ whose end sequent is $\Gamma \ra \Delta$. Notation  $ \mathsf{G}(\mathbf{L}) \vdash_{n} \Gamma \ra \Delta$ stands for the height of the derivation of that sequent.


\end{defn}

\begin{defn}
We define a well-ordered relation of sequent.

\begin{center}
$(\Gamma\ra \Delta)\prec (\Gamma^\prime\ra \Delta^\prime) $ if and only if 
$\mathsf{wt}(\Gamma,\Delta)   <
\mathsf{wt}(\Gamma^\prime,\Delta^\prime)$     
\end{center}
    
\end{defn}


By observing the weight of premises and conclusions in rules, we can obtain the following results.

\begin{prop}
\label{prop:termination of proof search in GKn}
    Backward proof-search in $\gkn$ and $\gkdn$ always terminates.
\end{prop}

In the next part, we will see that a backward proof-search in $\gktn $ does not always terminate. 
This fact requires us to provide
another sequent for $\gktn$ with loop-preventing mechanism.

\begin{defn}
    We say a rules is {\it admissible}, if for an instance of the rule, all premises are derivable, then there is derivation of its conclusion. 
    We say a rule is {\it height-preserving admissible} if for an instance of the rule, all premises are derivable with the greatest height $n$, then there is derivation of its conclusion with the height not greater than $n$. 

    We say a rule is  {\it height-preserving invertible} if for an instance of the rule, if the conclusion has a derivation with the height $n$, then each premise has a derivation with height not greater than $n$.
    
\end{defn}

In the  single-modal system from \cite{bilkova2007uniform}, the  context multiset (in this paper denoted by $\Sigma$) in the 
left part of the conclusion of a rule corresponding to $(\Box_{K_n})$ contains only propositional  variables.
However, in the setting of multi-modalities, in order to prove the admissibility of weakening rules, $\Sigma$ in $(\Box_{K_n})$ is permitted to contain also outmost-boxed (except the principal modality $\Box_i$) formulas.

For the details of the following proof of structural properties, please check \cite{Negri2001,troelstra2000basic,Kashima2009}.
\begin{prop}
Let $\mathbf{L}\in \{ \mathbf{K_n},\mathbf{KD_n}, \mathbf{KT_n}\}$.     For any formula $A$, a sequent $A\ra A$ is derivable in  $\mathsf{G}(\mathbf{L})$.
\end{prop}

\begin{prop}
\label{prop:weakening of k,d,t}
Let $\mathbf{L}\in \{ \mathbf{K_n},\mathbf{KD_n}, \mathbf{KT_n}\}$. The weakening rules are admissible in $\mathsf{G}(\mathbf{L})$.

\begin{center}
\begin{tabular}{ll} 

\AxiomC{$\Gamma \ra \Delta$}
\RightLabel{\scriptsize $(RW)$}
\UnaryInfC{$\Gamma \ra \Delta , C$}
\DisplayProof

   & 

\AxiomC{$\Gamma \ra \Delta$}
\RightLabel{\scriptsize $(LW)$}
\UnaryInfC{$C,\Gamma \ra \Delta$}
\DisplayProof

\end{tabular}
\end{center}

\end{prop}
\begin{proof}
    We proceed by double induction on weight of the formula  A and height of the derivation. 
 The admissibility of $(RW)$ can be easily obtained.
We only consider some cases where $A$ is in the form of $\Box_i B$ and the derivation ended with modal rule $(\Box_{K_n})$. Other cases are not difficult.
  
If $C$ is in the form $\Box_i A$ (i.e., $C$ is principal), we can obtain:

\begin{center}
\noLine
\AxiomC{$ I.H. $ }
\UnaryInfC{$C,\Gamma \ra B$}
\RightLabel{\scriptsize $(\Box_{Kn})$}
\UnaryInfC{$\Sigma, \Box i C, \Box i \Gamma   \ra \Box_i B, \Omega$}

\DisplayProof
\end{center}

If If $C$ is in the form $\Box_j A$ (i.e., $C$ is not principal),, we can  obtain:
\begin{center}
    
\AxiomC{$\Gamma \ra B$ }

\RightLabel{\scriptsize $(\Box_{Kn})$}
\UnaryInfC{$\Sigma, \Box j C, \Box i \Gamma   \ra \Box_i B, \Omega$}

\DisplayProof
\end{center}

\end{proof}
It is noted that the weakening rules are not height-preserving.

\begin{prop}
\label{prop:invertibility of all logical rules in Gkn}

       All logical rules in $\gkn$  except c are height-preserving invertible. 
       All logical rules in $\gkdn$ except $(\Box_{K_n})$ and $(\Box_{D_n})$ are height-preserving invertible. 
          All logical rules in $\gktn$ except $(\Box_{K_n})$ and $(\Box_{T_n})$ are height-preserving invertible. 

\end{prop}

Especially $(\Box_{T_n})$ cannot be height-preserving invertible in $\gktn$, since it does not satisfy  height-preserving weakening. 

\begin{prop}
Let $\mathbf{L}\in \{ \mathbf{K_n},\mathbf{KD_n}, \mathbf{KT_n}\}$.    
The contraction rules are height-preserving admissible in $\mathsf{G}(\mathbf{L})$.
    
  \begin{center}
\begin{tabular}{ll} 
 \AxiomC{$\Gamma \ra \Delta , A,A$}
\RightLabel{\scriptsize $(RC)$}
\UnaryInfC{$\Gamma \ra \Delta ,A $}
\DisplayProof

&

\AxiomC{$A,A,\Gamma \ra \Delta$}
\RightLabel{\scriptsize $(LC)$}
\UnaryInfC{$A,\Gamma \ra \Delta$}
\DisplayProof

\end{tabular}

    \end{center}

\end{prop}

\begin{proof}
    The proof is done  simultaneously by induction on the height of derivation of the premises.
    When the active formula $A$ is not principal in the end of derivation, the proof is straight.
   The case of $(LC) $ in  $(\Box_{K_n})$ is discussed as follows.
Suppose  that we have a derivation as:
 \begin{center}
\AxiomC{$\Gamma \ra B$}
\RightLabel{\scriptsize $(\Box_{Kn})$}
\UnaryInfC{$\Sigma,  A, A, \Box j \Gamma   \ra \Box_j B, \Omega$}
\DisplayProof
\end{center}
where the active formula $A$ can be a propositional variable , a $\bot$ or an outmost $i$-boxed formula such that $i\neq j$.
Then, we can provide a derivation as:
 \begin{center}
\AxiomC{$\Gamma \ra B$}
\RightLabel{\scriptsize $(\Box_{Kn})$}
\UnaryInfC{$\Sigma,  A, \Box j \Gamma   \ra \Box_j B, \Omega$}
\DisplayProof
\end{center}

    When the active formula $A$ is principal in the end of the derivation, Proposition \ref{prop:invertibility of all logical rules in Gkn} is needed for the cases of logical rules.
    Only the case of $(LC) $ in  $(\Box_{K_n})$ is discussed here, other modal rules can be proved similarly.

Suppose  that we have a derivation as:
 \begin{center}
\AxiomC{$A,A,\Gamma \ra B$}
\RightLabel{\scriptsize $(\Box_{Kn})$}
\UnaryInfC{$\Sigma, \Box i A,\Box i A, \Box i \Gamma   \ra \Box_i B, \Omega$}
\DisplayProof
\end{center}

Then, we can provide a derivation as:
 \begin{center}
 \noLine
\AxiomC{I.H.}
\UnaryInfC{$A,\Gamma \ra B$}
\RightLabel{\scriptsize $(\Box_{Kn})$}
\UnaryInfC{$\Sigma, \Box i A, \Box i \Gamma   \ra \Box_i B, \Omega$}
\DisplayProof
\end{center}
We apply the induction hypothesis to the premise of the assumption.
\end{proof}
It is noted that the repetition of a boxed formula in the premise of $(\Box_{Tn})$ is needed to prove the height-preserving admissibility of
contraction rules
\cite[Chapter 9.1]{troelstra2000basic}.

Next, the admissibility of cut rule is shown.
Since the weakening rules do not satisfy the height-preserving admissibility, 
we cannot directly apply the method for a standard G3-style sequent calculi for example in \cite[Theorem 4.15]{troelstra2000basic}. 
Here, the proof is done by following a similar argument in \cite[Theorem 3.23]{Negri2001}. Also, the cut rule is in the form without shared context. 

\begin{prop}
 
Let $\mathbf{L}\in \{ \mathbf{K_n},\mathbf{KD_n}, \mathbf{KT_n}\}$.     The cut rule is admissible in $\mathsf{G}(\mathbf{L})$.
    
\begin{center}


\begin{prooftree}
\AxiomC{$\vdots~\mathcal {D}_1$}
\RightLabel{$\mathtt{rule} (\mathcal{D}_1)$} 
\UnaryInfC{$\Gamma \ra  \Delta ,A$}
\AxiomC{$\vdots~\mathcal {D}_2$}
\RightLabel{$\mathtt{rule} (\mathcal{D}_2)$} 
\UnaryInfC{$ A , \Gamma ^{\prime}  \ra \Delta^{\prime} $}
\RightLabel{$(cut)$} 
\BinaryInfC{$\Gamma,\Gamma ^{\prime}   \ra \Delta, \Delta^{\prime}$}
\end{prooftree}
    
\end{center}
\begin{proof}

It is shown that if an $(cut)$  appears only in the end of a derivation $\mathcal{D}$, then there is a derivation in which no $(cut)$ appears and  ends with the same conclusion as $\mathcal{D}$. 
This can be proved by double induction on the {\it complexity} (the number of logical connectives of the cut formulas of $(cut)$) and the {\it height}, i.e., the number of all the sequents in the derivation.

The argument is divided into the following three cases: 
\begin{itemize}
\item[(1)] ${\mathcal D_1}$ or $ \mathcal{D}_2$ is an initial sequent.  
\item[(2)] $\mathtt{rule} (\mathcal{D}_1)$ or $\mathtt{rule} (\mathcal{D}_2)$ is a logical or modal rule in which the cut formula is not principal. 
\item[(3)] $\mathtt{rule} (\mathcal{D}_1)$ and $\mathtt{rule} (\mathcal{D}_2)$ are logical or modal rules, and the cut formulas are principal in both rules.
\end{itemize}

The proof proceeds basically  following an ordinary argument for the admissibility of cut 
for classical propositional logic \cite[Theorem 3.23]{Negri2001}. 

Only modal rules will be discussed here.
For the case (2), we only consider $\mathtt{rule} (\mathcal{D}_2)$ is $(\Box_{Kn})$ and cut formula is not principal.

\begin{center}
\noLine
\AxiomC{$\mathcal {D}_1$}
\UnaryInfC{$\Gamma\ra \Delta,C $}
\noLine
\AxiomC{$\mathcal {D}_2$}
\UnaryInfC{$\Gamma^{\prime} \ra B$}
\RightLabel{\scriptsize $(\Box_{Kn})$}
\UnaryInfC{$C,\Sigma,   \Box i \Gamma^{\prime}   \ra \Box_i B, \Omega$}
\RightLabel{\scriptsize $(cut)$}
\BinaryInfC{$\Sigma,   \Box i \Gamma^{\prime}, \Gamma   \ra \Box_i B, \Omega,\Delta$}
\DisplayProof
\end{center}

We can obtain the same result from the following derivation.
\begin{center}

\noLine
\AxiomC{$\mathcal {D}_2$}
\UnaryInfC{$\Gamma^{\prime} \ra B$}
\RightLabel{\scriptsize $(\Box_{Kn})$}
\UnaryInfC{$\Sigma,   \Box i \Gamma^{\prime}   \ra \Box_i B, \Omega$}
\RightLabel{\scriptsize $(weakening)^\ast$ }
\UnaryInfC{$\Sigma,   \Box i \Gamma^{\prime}, \Gamma   \ra \Box_i B, \Omega,\Delta$}
\DisplayProof
    
\end{center}
where  $(weakening)^\ast$ means applying weakening rules in
Proposition \ref{prop:weakening of k,d,t} for
finite many times.
The other case of $(\Box_{Kn})$ and
the cases of  $(\Box_{Dn})$, $(\Box_{Tn})$ can be similarly proved.

For the case (3), we consider $\mathtt{rule} (\mathcal{D}_1)$ is $(\Box_{Kn})$, $\mathtt{rule} (\mathcal{D}_2)$ is $(\Box_{Dn})$ and cut formulas are principal in both rules.

\begin{center}
\noLine
\AxiomC{$\mathcal {D}_1$}
\UnaryInfC{$\Gamma \ra C$}
\RightLabel{\scriptsize $(\Box_{Kn})$}
\UnaryInfC{$\Sigma,   \Box i \Gamma    \ra  \Omega,\Box_i C$}
\noLine
\AxiomC{$\mathcal {D}_2$}
\UnaryInfC{$C,\Gamma^{\prime} \ra  $}
\RightLabel{\scriptsize $(\Box_{Dn})$}
\UnaryInfC{$\Box iC,\Sigma^{\prime},   \Box i \Gamma^{\prime}   \ra   \Omega^{\prime}$}
\RightLabel{\scriptsize $(cut)$}
\BinaryInfC{$\Sigma, \Sigma^{\prime},  \Box_i\Gamma , \Box i \Gamma^{\prime}  \ra \Omega,\Omega^{\prime}$}
\DisplayProof
\end{center}

Then we can transform the derivation into the following:

\begin{center}
\noLine
\AxiomC{$\mathcal {D}_1$}
\UnaryInfC{$\Gamma \ra C$} 
\noLine
\AxiomC{$\mathcal {D}_2$}
\UnaryInfC{$C,\Gamma^{\prime} \ra  $} 
\RightLabel{\scriptsize $(cut)$}
\BinaryInfC{$\Gamma,\Gamma^{\prime} \ra$}
\RightLabel{\scriptsize $(\Box_{Dn})$}
\UnaryInfC{$\Sigma, \Sigma^{\prime},  \Box_i\Gamma , \Box i \Gamma^{\prime}  \ra \Omega,\Omega^{\prime}$}
\DisplayProof
\end{center}
In the transformed derivation, the application of $(cut)$ can be eliminated owing to the lower complexity of the cut formula.

In the case (3), when $\mathtt{rule} (\mathcal{D}_1)$ and $\mathtt{rule} (\mathcal{D}_2)$ are both   $(\Box_{Kn})$, also cut formulas are principal in both rules, the proof is similar to the above case.

For the case (3), $\mathtt{rule} (\mathcal{D}_1)$ is $(\Box_{Kn})$, $\mathtt{rule} (\mathcal{D}_2)$ is $(\Box_{Tn})$ and cut formulas are principal in both rules.

\begin{center}
\noLine
\AxiomC{$\mathcal {D}_1$}
\UnaryInfC{$\Gamma \ra C$}
\RightLabel{\scriptsize $(\Box_{Kn})$}
\UnaryInfC{$\Sigma,   \Box i \Gamma    \ra  \Omega,\Box_i C$}
\noLine
\AxiomC{$\mathcal {D}_2$}
\UnaryInfC{$\Box_i C, C,\Gamma^{\prime} \ra \Delta $}
\RightLabel{\scriptsize $(\Box_{Tn})$}
\UnaryInfC{$\Box iC ,  \Gamma^{\prime}   \ra   \Delta$}
\RightLabel{\scriptsize $(cut)$}
\BinaryInfC{$\Sigma,   \Box_i\Gamma ,  \Gamma^{\prime}  \ra \Delta,\Omega $}
\DisplayProof
\end{center}
We can transform the derivation into the following:
\begin{center}
\noLine
\AxiomC{$\mathcal {D}_1$}
\UnaryInfC{$\Gamma \ra C$}
\noLine
\AxiomC{$\mathcal {D}_1$}
\UnaryInfC{$\Sigma,   \Box i \Gamma    \ra  \Omega,\Box_i C$}
\noLine
\AxiomC{$\mathcal {D}_2$}
\UnaryInfC{$\Box_i C, C,\Gamma^{\prime} \ra \Delta $}
\RightLabel{\scriptsize $(cut)$}
\BinaryInfC{$C, \Sigma,   \Box i \Gamma,\Gamma^{\prime} \ra \Omega,\Delta$}
\RightLabel{\scriptsize $(cut)$}
\BinaryInfC{$ \Gamma, \Sigma,   \Box i \Gamma,\Gamma^{\prime} \ra \Omega,\Delta$}
\RightLabel{\scriptsize $(\Box_{Tn})^\ast$}
\UnaryInfC{$ \Sigma,   \Box i \Gamma,\Gamma^{\prime} \ra \Omega,\Delta$}
\DisplayProof
\end{center}
where  $(\Box_{Tn})^\ast$ denotes applying  $(\Box_{Tn})^\ast$ for finite many times.
In the transformed derivation, the uppermost application of $(cut)$ can be eliminated due to the reduced height of the derivation, while the second uppermost application of $(cut)$ can be eliminated owing to the lower complexity of the cut formula.
\end{proof}
It is noted that
the height-preserving admissibility of weakening is not necessary. The admissibility of weakening suffices to show the the proof.

\end{prop}
 
Next we show the equipollence between Hilbert systems and sequent calculi.

\begin{defn}
Given a sequent $\Gamma \ra \Delta$, $\Gamma_{\star}$ denotes the conjunction of all formulas in $\Gamma$ $(\Gamma_{\star}\equiv \top$ if $\Gamma$ is empty), $\Delta^{\star}$ denotes the unique formula in $\Delta$  ($\Delta^\star = \bot$ if $\Delta$ is empty).
\end{defn}

\begin{prop}
\label{prop:from g to h}
Let $\mathbf{L}\in \{ \mathbf{K_n},\mathbf{KD_n}, \mathbf{KT_n}\}$. 
 If $\mathsf{G}(\mathbf{L})   \vdash  \Gamma \ra \Delta$, then $\mathsf{H}(\mathbf{L})\vdash   \Gamma_{\star} \naraba \Delta ^\star$, 
\end{prop}

\begin{thm}[Equipollence]
\label{prop:equivilence of Hilbert and sequent}
Let $\mathbf{L}\in \{ \mathbf{K_n}, \mathbf{KD_n}, \mathbf{KT_n}\}$. 
The following equivalence holds: $\mathsf{H}(\mathbf{L}) \vdash        A$ iff  $ \mathsf{G}(\mathbf{L})  \vdash\ra A$. 
\end{thm}
\begin{proof}
The direction from the right to the left  can be proved by applying Proposition~\ref{prop:from g to h}, in which we let the antecedent $\Gamma$ be empty. 
The direction from the left to the right  can be proved by induction on the derivation of $A$. 
\end{proof}

Next, we move to prove our main results for $\gkn,\gkdn$. That of $\gktn$ will be shown after  defining  a sequent calculus with loop-preventing mechanism.

\subsection{Main theorem of $\gkn$ and $\gkdn$}
\label{sec:main thm of gkn and gkdn}

\begin{defn}
    We say that a sequent $\Gamma \ra \Delta$ is a {\it critical sequent} if $\Gamma$ and $\Delta$ contain only propositional variables or outmost-boxed formulas.
\end{defn}

\begin{defn}
\label{dfn:Ap formula in Gkn}
      Let $\Gamma,\Delta$ be finite multi-sets of formulas, $p$ be a propositional variable.
      An $\mathcal{A}$-formula $\mathcal{A}_p(\Gamma;\Delta)$  is defined inductively as follows.


\[
\begin{array}
{|c|c|c|c|}\hline & \Gamma;\Delta\text{ matches} & \mathcal{A}_p(\Gamma;\Delta)\text{ equals} \\
\hline


~~1~~ & \Gamma',p;\Delta',p & \top \\

~~2~~ & \Gamma',\bot;\Delta & \top  \\
~~3~~ & \Gamma',C_1\wedge C_2;\Delta & \mathcal{A}_p(\Gamma',C_1,C_2;\Delta) \\
4 & \Gamma;C_1\wedge C_2,\Delta' & \mathcal{A}_p(\Gamma;C_1,\Delta')\land \mathcal{A}_p(\Gamma;C_2,\Delta') \\

5 & \Gamma',C_1\vee C_2;\Delta & \mathcal{A}_p(\Gamma',C_1;\Delta)\land \mathcal{A}_p(\Gamma',C_2;\Delta) \\

~~6~~ & \Gamma;C_1\vee C_2,\Delta' & \mathcal{A}_p(\Gamma';C_1,C_2,\Delta') \\

~~7~~ & \Gamma',\neg C;\Delta & \mathcal{A}_p(\Gamma';C,\Delta) \\

8 & \Gamma;\neg C,\Delta' & \mathcal{A}_p(\Gamma,C;\Delta') \\

9 & \Gamma',C_1\rightarrow C_2;\Delta & \mathcal{A}_p(\Gamma';\Delta,C_1)\land \mathcal{A}_p(\Gamma',C_2;\Delta) \\
10 & \Gamma;C_1\rightarrow C_2,\Delta' & \mathcal{A}_p(\Gamma,C_1; \Delta',C_2) \\

11\dagger& ~~ \Phi, \overrightarrow{\Box_{g_m}\gamma_m } ; \overrightarrow{\Box_{d_n}\delta_n }, \Psi ~~  &  \mathsf{X} \\ 
 
\hline
\end{array}
\]
{\footnotesize $\dagger:$ $\Phi$ and  $\Psi$ are multisets containing only propositional variables or $\bot$, 
besides, $q,r$ differ from $p$. 
Furthermore $\Phi\cup \overrightarrow{\Box_{g_m}\gamma_m }\cup\overrightarrow{\Box_{d_n}\delta_n }\cup \Psi$ is not empty.}

The formula $\mathsf{X}$ is:

\[\bigvee_{q\in\Psi}q
\vee
\bigvee_{r\in\Phi}\neg r 
\vee 
\bigvee_{\Box_{g_j}\gamma_j \in \{\overrightarrow{\Box_{g_m}\gamma_m }\}} \Diamond_{g_j} \mathcal{A}_p (    \{ \overrightarrow{\Box_{g_m\gamma_m}} \}^{\flat_{g_j}} ;\emptyset) 
\]

\[
  \vee \bigvee_{\Box d_i\delta_i \in \{\overrightarrow{\Box_{d_n}\delta_n } \}}   \Box{d_i} \mathcal{A}_p (   \{\overrightarrow{\Box_{g_m}\gamma_m }  \}^{\flat_{d_i}}; \delta_i   )
\]


Recall that for any formula $A$, $n\in\mathbb{N}$, $\overrightarrow{A_n}$ stands for  $A_1,\cdots, A_n$. 
$\Gamma^{\flat_i}=\{A|\Box_iA\in \Gamma\}$,

The formula $\mathcal{A}_p(\Gamma;\Delta)$ is defined in the following procedure: at first, the lines $1-10$ are repeatedly applied until it reaches a critical sequent which does not match the line $1$ or $2$ (the order does not matter, since all propositional rules are height-preserving invertible by Proposition \ref{prop:invertibility of all logical rules in Gkn}); next, the line $11$ is applied. We repeat the above procedure until $\Gamma;\Delta$ cannot  match any lines in the table, in this case $\mathcal{A}_p(\Gamma;\Delta)$ is defined as $\bot$. Especially $\mathcal{A}_p(\emptyset;\emptyset)$ is defined as $\bot$.

We observe that for any ${\mathcal A}$-formulas defined in the right part, its weightiness always decrease when compare to its right part.
We can define a well-order relation of $\mathcal{A}$-formulas as follows:
\begin{center}
   $  \mathcal{A}_p (\Gamma;\Delta) \prec  \mathcal{A}_p (\Gamma^\prime;\Delta^\prime )$ if and only if

   $\mathsf{wt}(\Gamma\ra \Delta)  \prec  \mathsf{wt}(\Gamma^\prime\ra \Delta^\prime)$
\end{center}
Given the fact that all back proof-search in $\gkn$, $\gkdn$ always terminates (in Proposition \ref{prop:termination of proof search in GKn}), we can see that 
such a formula can always be determined.

\end{defn}

The definition in the single-modal setting from B{\'\i}lkov{\'a} \cite{bilkova2007uniform} becomes a particular case of this definition.

\begin{ex}
    Let $\Gamma$ be $\Box_1 (q\wedge p), \Box_2 (s\lor r)  ,\Box_2 r$,
    $\Delta$ be  $\Box_3 r, \Box_2 s$,
    then
    
    \begin{center}
    {\footnotesize
        $\mathcal{A}_p (
        \Box_1 (q\wedge p), \Box_2 (s\lor r)  ,\Box_2 r;\Box_3 r, \Box_2 s   )$
        
       $ =  
        \Diamond_1 \mathcal{A}_p 
        (q\wedge p;\emptyset) \lor
        \Diamond_2 \mathcal{A}_p(s\lor r, r;\emptyset) \lor
        \Diamond_2 \mathcal{A}_p(r,s\lor r; \emptyset) \lor
        \Box_3 \mathcal{A}_p (\emptyset; r ) \lor
         \Box_2 \mathcal{A}_p (s\lor r, r  ;s   )$

        $ =  
        \Diamond_1 \mathcal{A}_p 
        (q, p;\emptyset) \lor
    \Diamond_2 (\mathcal{A}_p(s, r;\emptyset) \wedge
          \mathcal{A}_p( r, r;\emptyset) )
        \lor
     \Diamond_2 (\mathcal{A}_p( r,s;\emptyset) \wedge
          \mathcal{A}_p( r, r;\emptyset) )
          \lor
        \Box_3 r \lor
         \Box_2 (\mathcal{A}_p (s,r;s) \wedge
         \mathcal{A}_p (r,r;s))$
        
                $ =  
       \Diamond_1( \neg q) \lor
               \Diamond_2 ((\neg s\lor \neg r) \wedge
          ( \neg r\lor \neg r) ) \lor
     \Diamond_2 (( \neg r\lor \neg s ) \wedge
          ( \neg r\lor \neg r ) )
          \lor
        \Box_3 r \lor
         \Box_2  ((\neg s \lor \neg r\lor s) \wedge
         (\neg r\lor\neg r \lor s))$

          }
    \end{center}
\end{ex}

We need the following proposition to deal with some cases in the main theorem.

\begin{prop}
\label{prop:initial casein main thm in gkn}
Let $\mathbf{L} \in \{ \mathbf{K_n},\mathbf{KD_n}\}$.  Given any multi-sets $\Gamma,\Delta$ of formulas, propositional variable $p$ and $q$, such that $p\neq q$.   
\begin{enumerate}
    \item $\mathsf{G}(\mathbf{L})\vdash q\ra \mathcal{A}_p(\Gamma;\Delta,q)$
    \item $\mathsf{G}(\mathbf{L})\vdash \neg q\ra \mathcal{A}_p(\Gamma,q;\Delta)$
    \item $\mathsf{G}(\mathbf{L})\vdash \ra \mathcal{A}_p(q,\Gamma;\Delta,q)$
\end{enumerate}\end{prop}

\begin{thm}
\label{thm:main theorem of gkn}
Let $\mathbf{L}\in \{ \mathbf{K_n},\mathbf{KD_n}\}$.
    Let $\Gamma,\Delta$ be finite multi-sets of formulas. For every propositional variable $p$ there exists a formula $\mathcal{A}_p(\Gamma;\Delta)$ such that:
    \begin{enumerate}[(i)]
        \item  
            $\mathsf{V}( \mathcal{A}_p(\Gamma;\Delta)) \subseteq \mathsf{V}(\Gamma\cup \Delta)\backslash\{p\}    $
        \item $\mathsf{G}(\mathbf{L})\vdash    \Gamma , \mathcal{A}_p(\Gamma; \Delta) \ra \Delta$
        \item given finite multi-sets $\Pi, \Lambda$ of formulas such that 
        \begin{center}
            $p\notin \mathsf{V}(\Pi\cup\Lambda)$ and $\mathsf{G}(\mathbf{L})\vdash  \Pi,\Gamma \ra \Delta, \Lambda$
        \end{center}
        then 
        \begin{center}
            $\mathsf{G}(\mathbf{L})\vdash  \Pi \ra  \mathcal{A}_p(\Gamma; \Delta) ,\Lambda$
        \end{center}
    \end{enumerate}
    
\end{thm}

\begin{proof}


The proof of (i) can be obtained by inspecting the table in Definition \ref{dfn:Ap formula in Gkn}.

The proof of  (ii) can be proved by induction on the weight of $\mathcal{A}_p(\Gamma;\Delta)$. We prove $\mathsf{G}(\mathbf{L})\vdash    \Gamma , \mathcal{A}_p(\Gamma; \Delta) \ra \Delta$ for each line of the table in  Definition \ref{dfn:Ap formula in Gkn}.
The cases of lines from 1 to 10 are easy.
We only consider the case of the line 11.

The idea is to show the results of each conjuncts at first, then we combine them together by $(L\lor)$.

\begin{itemize}
    \item for each $q\in \Psi$, $ \mathsf{G}(\mathbf{L}) \vdash \Phi,q, \overrightarrow{\Box_{g_m}\gamma_m }  \ra  \overrightarrow{\Box_{d_n}\delta_n },\Psi $

    \item for each $r\in \Phi$, $ \mathsf{G}(\mathbf{L}) \vdash \Phi,\neg r, \overrightarrow{\Box_{g_m}\gamma_m }  \ra  \overrightarrow{\Box_{d_n}\delta_n },\Psi $

    \item for each  $\Box_{g_j}\gamma_j \in \{\overrightarrow{\Box_{g_m}\gamma_m }\}$,

$ \mathsf{G}(\mathbf{L}) \vdash    \{ \overrightarrow{\Box_{g_m\gamma_m}} \}^{\flat_{gj}},\mathcal{A}_p (    \{ \overrightarrow{\Box_{g_m\gamma_m}} \}^{\flat_{gj}} ;\emptyset)  \ra    $  from induction hypothesis. After applying $(R\neg)$, $(\Box_{Kn})$,$(L\neg)$  we obtain

$ \mathsf{G}(\mathbf{L}) \vdash \Phi, \Diamond_{g_j} \mathcal{A}_p (    \{ \overrightarrow{\Box_{g_m\gamma_m}} \}^{\flat_{gj}} ;\emptyset),\overrightarrow{\Box_{g_m}\gamma_m }  \ra  \overrightarrow{\Box_{d_n}\delta_n },\Psi   $ 
\item  
for each $\Box d_i\delta_i \in \{\overrightarrow{\Box_{d_n}\delta_n } \}$, by induction hypothesis, we obtain:
 
$  \mathsf{G}(\mathbf{L}) \vdash \{\overrightarrow{\Box_{g_m}\gamma_m }  \}^{\flat_{di}}, \mathcal{A}_p (   \{\overrightarrow{\Box_{g_m}\gamma_m }  \}^{\flat_{di}}; \delta_i   )  
\ra \delta_i  .$
After applying $ (\Box_{Kn})$,   we obtain

$ \mathsf{G}(\mathbf{L}) \vdash\Box{d_i} \mathcal{A}_p (   \{\overrightarrow{\Box_{g_m}\gamma_m }  \}^{\flat_{di}}; \delta_i   ) ,\overrightarrow{\Box_{g_m}\gamma_m } , \Phi
\ra \overrightarrow{\Box_{d_n}\delta_n },\Psi $

\end{itemize}

After applying $(L\lor)$ finite many times to the above results, we obtain:

\[
 \mathsf{G}(\mathbf{L}) \vdash
\bigvee_{q\in\Psi}q
\vee
\bigvee_{r\in\Phi}\neg r 
\vee 
\bigvee_{\Box_{g_j}\gamma_j \in \{\overrightarrow{\Box_{g_m}\gamma_m }\}} \Diamond_{g_j} \mathcal{A}_p (    \{ \overrightarrow{\Box_{g_m\gamma_m}} \}^{\flat_{gj}} ;\emptyset) 
\]

\[
  \vee \bigvee_{\Box d_i\delta_i \in \{\overrightarrow{\Box_{d_n}\delta_n } \}}   \Box{d_i} \mathcal{A}_p (   \{\overrightarrow{\Box_{g_m}\gamma_m }  \}^{\flat_{di}}; \delta_i   ) ,\overrightarrow{\Box_{g_m}\gamma_m } , \Phi
\ra \overrightarrow{\Box_{d_n}\delta_n },\Psi 
\]


In (iii), we consider the last rules applied in the derivation of  $\vdash_{\gkdn} \Pi,\Gamma \ra \Delta, \Lambda$.

 When it is an initial sequent, we need the Proposition \ref{prop:initial casein main thm in gkn}.
When the last rules are logical rules, the proof is straight.

When the last rule is $(\Box_{K_n})$, we have a derivation of 
$\Pi, \Gamma \ra \Delta, \Lambda$.
Then the arguments are divided into the following cases:
\begin{enumerate}
    \item The right principal formula $\Box_i A$ is in the multiset $\Lambda$.

\begin{center}
    
\AxiomC{$\Pi'',\Gamma''\ra A$}
\RightLabel{$(\Box_{K_n}) $}
\UnaryInfC{$\Pi', \Box_i \Pi'',\Phi,\overrightarrow{\Box g_m \gamma _m}, \Box_i \Gamma'' \ra \Box_i A, \Lambda', \overrightarrow{\Box d_n \delta_n}, \Psi $}
\DisplayProof

\end{center}
where

\begin{itemize}
        \item $\Pi', \Box_i \Pi''$ is $\Pi$, and $\Pi'$ contains propositional variables, $\bot$ and outmost-boxed formula except $\Box_i$;
    \item $\Phi,\overrightarrow{\Box g_m \gamma_m}, \Box_i \Gamma''$ is $\Gamma$, where $\Box_i$ is not among  $\overrightarrow{\Box g_m  }$, $\Phi$ contains only propositional variables, $\bot$;
    \item $\Box_i A, \Lambda'$ is $\Lambda$, and $\Lambda'$ contains propositional variables, $\bot$ and outmost-boxed formula; 
    \item $ \overrightarrow{\Box d_n \delta_n}, \Psi$ is $\Delta$, and $\Psi$ contains only propositional variables, $\bot$.
\end{itemize}   
    
    \begin{enumerate}
        \item Some formulas in $\Gamma$ are principal.

        In this case, the above $\Gamma''$ is non-empty. From the assumption $p\notin \mathsf{V}(\Pi\cup\Lambda)$, then $p\notin \mathsf{V}(\Pi''\cup \{A\})$. Then by induction hypothesis, 
 \begin{center}
        \AxiomC{$\Pi''\ra \mathcal{A}_p( \Gamma'';\emptyset   ), A$}
          \RightLabel{\footnotesize $( L\neg) $}
     \UnaryInfC{{$\Pi'', \neg \mathcal{A}_p( \Gamma'';\emptyset   )\ra  A$}}
     \RightLabel{\footnotesize $(\Box_{K_n}) $}
     \UnaryInfC{{$\Pi',\Box_i\Pi'', \Box_i \neg \mathcal{A}_p( \Gamma'';\emptyset   ),\ra  \Box_iA$}}
     \RightLabel{\footnotesize $(R\neg ) $}
     \UnaryInfC{{$\Pi',\Box_i\Pi'' \ra \Diamond_i  \mathcal{A}_p( \Gamma'';\emptyset   ), \Box_iA$}}
          \RightLabel{\footnotesize $(RW ),(R\lor ) $}
     \UnaryInfC{ }
     \DisplayProof   
     \[\Pi',\Box_i\Pi'' \ra 
     \bigvee_{ \Box_i\gamma \in  \{\overrightarrow{\Box g_m \gamma_m}, \Box_i \Gamma''  \}}
     \Diamond_i
     \mathcal{A}_p (    \{ \overrightarrow{\Box_{g_m\gamma_m}}, \Box_i\Gamma''
     \}^{\flat_{i}} ;\emptyset) , \Box_i A
     \]
 \end{center}
 Then we obtain the derivation of $\Pi \ra \mathcal{A}_p(\Gamma;\Delta),\Lambda$ from applying $(R\lor) $ and weakening rules for finitely many times.
 
 It is remarked that, $  \{ \overrightarrow{\Box_{g_m\gamma_m}}, \Box_i\Gamma''    \}^{\flat_{i}} = \Gamma''$ since  $\Box_i$ is not among  $\overrightarrow{\Box g_m  }$.
        
        \item All formulas in $\Gamma$ are not principal. That is, $\Gamma''=\emptyset$. From the assumption, we derive:

         \begin{center}
        \AxiomC{$\Pi''\ra   A$} 
     \RightLabel{\footnotesize $(\Box_{K_n}) $}
     \UnaryInfC{{$\Pi'',\Box_i\Pi'' \ra  \Box_iA, \Lambda'$}}
     \DisplayProof
 \end{center}
 Then we obtain the derivation of $\Pi \ra \mathcal{A}_p(\Gamma;\Delta),\Lambda$ from applying   weakening rules for finitely many times.
        
    \end{enumerate}
    
    \item The right principal formula $\Box_i A$ is in the multiset $\Delta$.
       \begin{enumerate}
        \item Some formulas in $\Gamma$ are principal.

\begin{center}
    
\AxiomC{$\Pi'',\Gamma''\ra A$}
\RightLabel{\footnotesize $(\Box_{K_n}) $}
\UnaryInfC{$\Pi', \Box_i \Pi'',\Phi,\overrightarrow{\Box g_m \gamma _m}, \Box_i \Gamma'' \ra \Box_i A, \overrightarrow{\Box d_n \delta_n}, \Psi, \Lambda $}
\DisplayProof

\end{center}
where

\begin{itemize}
        \item $\Pi', \Box_i \Pi''$ is $\Pi$, and $\Pi'$ contains propositional variables, $\bot$ and outmost-boxed formula except $\Box_i$;
    \item $\Phi,\overrightarrow{\Box g_m \gamma_m}, \Box_i \Gamma''$ is $\Gamma$, where $\Box_i$ is not among  $\overrightarrow{\Box g_m  }$, $\Phi$ contains only propositional variables, $\bot$;
    \item  $\Lambda$ contains propositional variables, $\bot$ and outmost-boxed formula; 
    \item $ \Box_i A, \overrightarrow{\Box d_n \delta_n}, \Psi$ is $\Delta$, and $\Psi$ contains only propositional variables, $\bot$.
\end{itemize}   

        In this case, the above $\Gamma''$ is non-empty. From the assumption $p\notin \mathsf{V}(\Pi)$, then $p\notin \mathsf{V}(\Pi'')$. Then by induction hypothesis, 
 \begin{center}
        \AxiomC{$\Pi''\ra \mathcal{A}_p( \Gamma'';A   )$}
     \RightLabel{\footnotesize $(\Box_{K_n}) $}
     \UnaryInfC{{$\Pi',\Box_i\Pi''\ra  \Box_i  \mathcal{A}_p( \Gamma'';A   ),\Lambda$}}

          \RightLabel{\footnotesize $(RW ), (R\lor ) $}
     \UnaryInfC{     }
     \DisplayProof

      \[\Pi',\Box_i\Pi'' \ra 
     \bigvee_{\Box_i A \in  \{ \Box_i A, \overrightarrow{\Box d_n \delta_n}   \} }
     \Box_i
     \mathcal{A}_p (    \{ \overrightarrow{\Box_{g_m\gamma_m}}, \Box_i\Gamma''
     \}^{\flat_{i}} ;A) , \Lambda
     \] 
 \end{center}

 Then we obtain the derivation of $\Pi \ra \mathcal{A}_p(\Gamma;\Delta),\Lambda$ from applying $(R\lor) $ and weakening rules for finitely many times.

  It is remarked that, $   \{ \overrightarrow{\Box_{g_m\gamma_m}}, \Box_i\Gamma''
     \}^{\flat_{i}}  = \Gamma''$
  since  $\Box_i$ is not among  $\overrightarrow{\Box g_m  }$.
 
        \item All formulas in $\Gamma$ are not principal. Then, $\Gamma''$ is empty.

From the assumption $p\notin \mathsf{V}(\Pi)$, then $p\notin \mathsf{V}(\Pi'')$. Then by induction hypothesis, 
 \begin{center}
        \AxiomC{$\Pi''\ra \mathcal{A}_p( \emptyset;A   )$}
     \RightLabel{\footnotesize $(\Box_{K_n}) $}
     \UnaryInfC{{$\Pi',\Box_i\Pi''\ra  \Box_i  \mathcal{A}_p( \emptyset;A   ),\Lambda$}}

          \RightLabel{\footnotesize $(RW ), (R\lor ) $}
     \UnaryInfC{     }
     \DisplayProof

      \[\Pi',\Box_i\Pi'' \ra 
     \bigvee_{\Box_i A \in  \{ \Box_i A, \overrightarrow{\Box d_n \delta_n}   \} }
     \Box_i
     \mathcal{A}_p (    \{ \overrightarrow{\Box_{g_m\gamma_m}}
     \}^{\flat_{i}} ;A) , \Lambda
     \] 
 \end{center}
  
 Then we obtain the derivation of $\Pi \ra \mathcal{A}_p(\Gamma;\Delta),\Lambda$ from applying $(R\lor) $ and weakening rules for finitely many times.

  It is remarked that, $   \{ \overrightarrow{\Box_{g_m\gamma_m}}
     \}^{\flat_{i}}  = \emptyset$,
  since $\Gamma''$ is empty and $\Box_i$ is not among  $\overrightarrow{\Box g_m  }$.

    \end{enumerate}
\end{enumerate}

When the last rule is $(\Box_{D_n})$, the derivation of $\Gamma,\Pi\ra \Delta,\Lambda$ is in the form of 
\begin{center}
    
\AxiomC{$\Pi'',\Gamma''\ra $}
\RightLabel{\footnotesize $(\Box_{D_n}) $}
\UnaryInfC{$\Pi', \Box_i \Pi'',\Phi,\overrightarrow{\Box g_m \gamma _m}, \Box_i \Gamma'' \ra \Delta, \Lambda$}
\DisplayProof

\end{center}
where
\begin{itemize}
        \item $\Pi', \Box_i \Pi''$ is $\Pi$, and $\Pi'$ contains propositional variables, $\bot$ and outmost-boxed formula except $\Box_i$;
    \item $\Phi,\overrightarrow{\Box g_m \gamma_m}, \Box_i \Gamma''$ is $\Gamma$, where $\Box_i$ is not among  $\overrightarrow{\Box g_m  }$, $\Phi$ contains only propositional variables, $\bot$;
    \item  $\Lambda$ contains propositional variables, $\bot$ and outmost-boxed formula; 
    \item  $\Delta$ contains propositional variables, $\bot$ and outmost-boxed formula.
\end{itemize}   

Our arguments are divided into the following cases:
\begin{enumerate}

    \item  Some formulas in  $\Gamma$ are principal .

        In this case, the above $\Gamma''$ is non-empty. From the assumption $p\notin \mathsf{V}(\Pi)$, then $p\notin \mathsf{V}(\Pi'')$. Then by induction hypothesis, 
        
 \begin{center}
 
        \AxiomC{$\Pi''\ra \mathcal{A}_p( \Gamma'';\emptyset  )$}
        \RightLabel{\footnotesize $(L\neg)$}
        \UnaryInfC{$\Pi'',\neg \mathcal{A}_p( \Gamma'';\emptyset  )\ra $}
     \RightLabel{\footnotesize $(\Box_{D_n}) $}
     \UnaryInfC{{$\Pi',\Box_i\Pi'', \Box_i \neg  \mathcal{A}_p( \Gamma'';\emptyset  )\ra \Lambda$}}
     \RightLabel{\footnotesize ($R\neg$)}
  \UnaryInfC{{$\Pi',\Box_i\Pi'' \ra \Diamond_i  \mathcal{A}_p( \Gamma'';\emptyset  ),\Lambda$}}
          \RightLabel{\footnotesize $(RW ), (R\lor ) $}
     \UnaryInfC{     }
     \DisplayProof
      \[\Pi',\Box_i\Pi'' \ra 
       \bigvee_{ \Box_i\gamma \in  \{\overrightarrow{\Box g_m \gamma_m}, \Box_i \Gamma''  \}}
     \Diamond_i
     \mathcal{A}_p (    \{ \overrightarrow{\Box_{g_m\gamma_m}}, \Box_i\Gamma''
     \}^{\flat_{i}} ;\emptyset) , 
     \Lambda
     \] 
 \end{center}

 Then we obtain the derivation of $\Pi \ra \mathcal{A}_p(\Gamma;\Delta),\Lambda$ from applying $(R\lor) $ and weakening rules for finitely many times.

  It is remarked that, $   \{ \overrightarrow{\Box_{g_m\gamma_m}}, \Box_i\Gamma''
     \}^{\flat_{i}}  = \Gamma''$
  since  $\Box_i$ is not among  $\overrightarrow{\Box g_m  }$.

    \item  None of formulas  in $\Gamma$ are principal.

 All formulas in $\Gamma$ are not principal. That is, $\Gamma''=\emptyset$. From the assumption, we derive:

         \begin{center}
        \AxiomC{$\Pi''\ra   $} 
     \RightLabel{\footnotesize $(\Box_{K_n}) $}
     \UnaryInfC{{$\Pi',\Box_i\Pi'' \ra    \Lambda$}}
     \DisplayProof
 \end{center}
 Then we obtain the derivation of $\Pi \ra \mathcal{A}_p(\Gamma;\Delta),\Lambda$ from applying   weakening rules for finitely many times. \qedhere
\end{enumerate}
\end{proof}

\begin{defn}
    In language $\mathcal{L}^1$, let $p$ be a propositional variable and $B$ be a formula. We define $\mathcal{A}_p(B)$ as $\mathcal{A}_p(\emptyset;B)$.
    Furthermore, we define $\mathcal{E}_p(B)$ as $\neg\mathcal{A}_p(\neg B)$, namely $\neg \mathcal{A}_p(\emptyset;\neg B)$.
\end{defn}

\begin{cor}
\label{cor:uip in kn, kdn, first}
    Uniform interpolation properties are satisfied in $\gkn$ and $\gkdn$ language $\mathcal{L}^1$.

 Let $\mathbf{L}\in \{ \mathbf{K}_n, \mathbf{KD}_n\}$.
For any formula $B(\overrightarrow{q},\overrightarrow{r})$, such that all $q$ are different from all $r$,
 there exists a formula (pre-interpolant) $\mathcal{I}_{pre}(B, \overrightarrow{q} )$  such that:
\begin{enumerate}
    \item all $ \overrightarrow{r}$ do not occur in $\mathcal{I}_{pre}(B, \overrightarrow{q} )$;
    \item $\mathcal{I}_{pre} (B, \overrightarrow{q} ) \ra B(\overrightarrow{q},\overrightarrow{r})  $ is derivable in $\mathsf{G}(\mathbf{L})$;
    \item  for any formula $A(\overrightarrow{p},\overrightarrow{q})$, where
    all $p$ are different from all $q$, 
    if $A(\overrightarrow{p},\overrightarrow{q})  \ra  B(\overrightarrow{q},\overrightarrow{r})$ is derivable in $\mathsf{G}
    (\mathbf{L})$ then   $
    A(\overrightarrow{p},\overrightarrow{q})\ra
\mathcal{I}_{pre}(B,\overrightarrow{q} )
    $ is derivable in $\mathsf{G}(\mathbf{L})$.
\end{enumerate}

   Furthermore,  for any formula $A(\overrightarrow{p},\overrightarrow{q})$, such that all $q$ are different from all $p$,
 there exists a formula (post-interpolant) $\mathcal{I}_{post}(A, \overrightarrow{q} )$  such that:
\begin{enumerate}
   \item all $\overrightarrow{p}$ do not occur in $\mathcal{I}_{post}(A, \overrightarrow{q} )$; 
    \item $A(\overrightarrow{p},\overrightarrow{q})\ra\mathcal{I}_{post} (A, \overrightarrow{q} )  $ is derivable in $\mathsf{G}(\mathbf{L})$;
    \item  for any formula $B(\overrightarrow{q},\overrightarrow{r})$, where
    all $r$ are different from all $q$, 
    if $A(\overrightarrow{p},\overrightarrow{q})  \ra  B(\overrightarrow{q},\overrightarrow{r})$ is derivable in $\mathsf{G}
    (\mathbf{L})$ then   $\mathcal{I}_{post}(A,\overrightarrow{q} )\ra B(\overrightarrow{q},\overrightarrow{r}) $ is derivable in $\mathsf{G}
    (\mathbf{L})$.
\end{enumerate}
\end{cor}

\begin{proof}
First, we show the case of pre-interpolatant. Given an arbitrary formula $B(\overrightarrow{q},\overrightarrow{r})$.
Let $\mathcal{I}_{pre} (B, \overrightarrow{q} ) $ be $\mathcal{A}r_1 ( \mathcal{A}r_2(   \dots (\mathcal{A}r_m (B(\overrightarrow{q},\overrightarrow{r}))) \dots ))$, where $r_1,\dots,r_m $ stands for $\overrightarrow{r}$.
Then, according to Definition \ref{dfn:Ap formula in Gkn} we can easily check that $\{r_1,\dots,r_m \} \cap \mathsf{V} ( \mathcal{I}_{pre} (B, \overrightarrow{q} )   ) =\emptyset $.
Item 2 and 3 can be easily proved from (ii),(iii) of Theorem \ref{thm:main theorem of gkn}. 

Next, we show the case of post-interpolant.  Given an arbitrary formula $A(\overrightarrow{p},\overrightarrow{q})$.
Let $\mathcal{I}_{post} (A, \overrightarrow{q} ) $ be $\mathcal{E}p_1 ( \mathcal{E}p_2(   \dots (\mathcal{E}p_m (A(\overrightarrow{p},\overrightarrow{q}))) \dots ))$, where $p_1,\dots,p_n $ stands for $\overrightarrow{p}$.
Similarly, we can check that $\{p_1,\dots,p_m \} \cap \mathsf{V} ( \mathcal{I}_{post} (A, \overrightarrow{q} )   ) =\emptyset $.
According to (ii) of Theorem \ref{thm:main theorem of gkn}, 
$\mathsf{G}(\mathbf{L})  \vdash  \mathcal{A}p_m ( \neg A(\overrightarrow{p},\overrightarrow{q})) \ra   \neg A(\overrightarrow{p},\overrightarrow{q})$.
Then we obtain $\mathsf{G}(\mathbf{L})\vdash A(\overrightarrow{p},\overrightarrow{q})  \ra \neg \mathcal{A}p_m ( \neg A(\overrightarrow{p},\overrightarrow{q})) $.
Since $\mathsf{G}(\mathbf{L})\vdash  \mathcal{A}p_{m-1} (\mathcal{A}p_m ( \neg A(\overrightarrow{p},\overrightarrow{q}))) \ra \mathcal{A}p_m ( \neg A(\overrightarrow{p},\overrightarrow{q})) $, we have $\mathsf{G}(\mathbf{L})\vdash    \neg \mathcal{A}p_m ( \neg A(\overrightarrow{p},\overrightarrow{q})) \ra \neg \mathcal{A}p_{m-1} (\mathcal{A}p_m ( \neg A(\overrightarrow{p},\overrightarrow{q})))$.
By repeating this argument  and applying cut rules, we can  obtain 
$\mathsf{G}(\mathbf{L})\vdash    A(\overrightarrow{p},\overrightarrow{q}) \ra \mathcal{E}p_1 ( \mathcal{E}p_2(   \dots (\mathcal{E}p_m (A(\overrightarrow{p},\overrightarrow{q}))) \dots ))$.

Item 3 of post-interpolant can be shown by (iii) of  Theorem \ref{thm:main theorem of gkn}.
From assumption, we can easily derive $\mathsf{G}(\mathbf{L})\vdash  \ra   B(\overrightarrow{q},\overrightarrow{r}), \neg A(\overrightarrow{p},\overrightarrow{q})$.
Then, we derive $\mathsf{G}(\mathbf{L})\vdash  \ra   B(\overrightarrow{q},\overrightarrow{r}), \mathcal{A}{p_n} ( \neg A(\overrightarrow{p},\overrightarrow{q}))$. Then, by a similar argument, we can derive that $\mathsf{G}(\mathbf{L})\vdash  \ra   B(\overrightarrow{q},\overrightarrow{r}), \mathcal{A}{p_{n-1}}(\mathcal{A}{p_n} ( \neg A(\overrightarrow{p},\overrightarrow{q})))$. After repeating this argument, we can obtain  $\mathsf{G}(\mathbf{L})\vdash \mathcal{E}p_1 ( \mathcal{E}p_2(   \dots (\mathcal{E}p_m (A(\overrightarrow{p},\overrightarrow{q}))) \dots )) \ra  B(\overrightarrow{q},\overrightarrow{r})$ by $(L\neg)$.
\end{proof}

\section{UIP in Logic $\mathbf{ KT_n}$}

\subsection{Proof-theoretic properties of sequent calculus}
\label{sec:sequent cal of T}












 





We can observe the following example to find that a backward-proof search may not terminate in $\gktn$.

\begin{ex}
    The backward proof-searching of the sequent $p\Rightarrow\Diamond_i (p\wedge q)$ does not terminate (cf. \cite{bilkova2007uniform}). 
\end{ex}

\begin{center}

\AxiomC{$p, \Box_i\neg(p\wedge q) \ra p $}
\AxiomC{$ p, \Box_i\neg(p\wedge q) \ra p,q  $}

\AxiomC{$ \vdots$}
\UnaryInfC{$ p, \Box_i\neg(p\wedge q) \ra q,q$}
\RightLabel{\scriptsize $(R\wedge )$}
\BinaryInfC{$p, \Box_i\neg(p\wedge q) \ra p\wedge q,q$}
\RightLabel{\scriptsize $(L\neg)$}
\UnaryInfC{$p, \neg (p\wedge q),\Box_i\neg(p\wedge q) \ra q$}
\RightLabel{\scriptsize $(\Box_{T_n})$}
\UnaryInfC{$p, \Box_i\neg(p\wedge q) \ra q$}

\RightLabel{\scriptsize ($R\wedge $)}
\BinaryInfC{ $p, \Box_i\neg(p\wedge q) \ra p\wedge q$  }
\RightLabel{\scriptsize $(L\neg )$}
\UnaryInfC{$p, \neg (p\wedge q),\Box_i\neg(p\wedge q) \ra $  }
\RightLabel{\scriptsize $(\Box_{T_n})$}
\UnaryInfC{$p, \Box_i\neg(p\wedge q) \ra $  }
\RightLabel{\scriptsize $(R\neg)$}
\UnaryInfC{$p \ra \neg \Box_i\neg(p\wedge q) $  }
\DisplayProof
    
\end{center}

This possible loop in derivation of $\gktn$ will bring difficulties in defining $\mathcal{A}$-formula.
In the following, a sequent calculus with
built-in loop-check mechanism will be presented.
This sequent is an expansion of single modal calculus from   \cite{bilkova2007uniform}.

A {\bf T}-sequent, denoted by $\Sigma|\Gamma \ra \Delta$ is obtained from adding a finite multiset $\Sigma$  into a sequent $\Gamma\ra \Delta$, where $\Sigma$ containing only outmost-boxed formulas,

\begin{table}[htb]
\caption{System $\gktnplus$.}
\label{table:gktnplus}
\begin{center}
\hrule
\begin{tabular}{ll}

\multicolumn{2}{c}{ Sequent Calculus $\gktnplus$: } \\


{ \bf Initial Sequents}   & $\Sigma|\Gamma ,p \ra p, \Delta$\hspace{15pt} $\Sigma|\bot, \Gamma\ra \Delta$  \\ 
 \\







{\bf Logical Rules} & 

\AxiomC{$\Sigma|\Gamma \ra \Delta ,A_1$}
\AxiomC{$\Sigma|\Gamma \ra \Delta ,A_2$}
\RightLabel{\scriptsize $(R\wedge )$}
\BinaryInfC{$\Sigma|\Gamma \ra  \Delta ,A_1\wedge A_2$}
\DisplayProof

\AxiomC{$\Sigma|A_1, A_2,\Gamma \ra \Delta$}
\RightLabel{\scriptsize $(L\wedge )$}
\UnaryInfC{$\Sigma|A_1\wedge A_2,\Gamma \ra \Delta$}
\DisplayProof
\\

\,  &

\AxiomC{$\Sigma|\Gamma \ra A_1, A_2$}
\RightLabel{\scriptsize $(R\lor )$}
\UnaryInfC{$\Sigma|\Gamma \ra A_1\lor A_2$}
\DisplayProof

\AxiomC{$\Sigma|A_1,\Gamma \ra \Delta$}
\AxiomC{$\Sigma|A_2,\Gamma \ra \Delta$}
\RightLabel{\scriptsize $(L\lor )$}
\BinaryInfC{$\Sigma|A_1\lor A_2, \Gamma \ra \Delta$}
\DisplayProof

\\

\,  &

\AxiomC{$\Sigma|A_1,\Gamma \ra \Delta, A_2$}
\RightLabel{\scriptsize $(R\rightarrow )$}
\UnaryInfC{$\Sigma|\Gamma \ra \Delta,A_1\rightarrow A_2$}
\DisplayProof

\AxiomC{$\Sigma|\Gamma \ra \Delta, A$}
\AxiomC{$\Sigma|A_2,\Gamma \ra \Delta$}
\RightLabel{\scriptsize $(L\rightarrow )$}
\BinaryInfC{$\Sigma|A_1\rightarrow A_2, \Gamma \ra \Delta$}
\DisplayProof

\\

\,  &
\AxiomC{$\Sigma|A, \Gamma \ra  \Delta$}
\RightLabel{\scriptsize $(R\neg )$}
\UnaryInfC{$\Sigma|\Gamma \ra \Delta,\neg A$}
\DisplayProof

\AxiomC{$\Sigma|\Gamma \ra \Delta, A$}
\RightLabel{\scriptsize $(L\neg )$}
\UnaryInfC{$\Sigma|\neg A, \Gamma \ra  \Delta$}
\DisplayProof
\\

{\bf Modal Rule}  &

\AxiomC{$ \emptyset | \Gamma \ra A $ }
\RightLabel{\scriptsize $(\Box_{Kn}^+)$$\dagger$}
\UnaryInfC{$ \Sigma, \Box_i\Gamma | \Pi  \ra \Box_iA, \Omega$}
\DisplayProof
\,
\AxiomC{$ \Box_i A, \Sigma|  \Gamma, A \ra \Delta $}
\RightLabel{\scriptsize $(\Box_{Tn}^+)$}
\UnaryInfC{$\Sigma |\Gamma, \Box i A    \ra \Delta$}
\DisplayProof

\\

\multicolumn{2}{l}{ {\footnotesize $\dagger$: $\Sigma$ contains only outmost-boxed formulas except $\Box_i$, $\Pi$
contains only propositional variables and  $\bot$.}
   } \\
   
\multicolumn{2}{l}{ {\footnotesize $\Omega$ contains only propositional variables, $\bot$ or outmost-boxed formulas. }
   } 
   \\

\hline

\end{tabular}
\end{center}
\end{table}

\begin{defn}
The system of $\gktnplus$  is defined in Table \ref{table:gktnplus}.

\end{defn}

\begin{ex}
    The backward proof-searching of the sequent  $\emptyset|p\Rightarrow\Diamond_i (p\wedge q)$  terminates.
\end{ex}

\begin{center}

\AxiomC{$ \Box_i\neg(p\wedge q) | p \ra p  $}
\AxiomC{$ \Box_i\neg(p\wedge q) | p \ra q$} 

\RightLabel{\scriptsize $(R\wedge )$}
\BinaryInfC{$\Box_i\neg(p\wedge q) | p \ra p\wedge q$}
\RightLabel{\scriptsize $(L\neg)$}
\UnaryInfC{$\Box_i\neg(p\wedge q) | p ,\neg (p\wedge q)\ra $}
\RightLabel{\scriptsize $(\Box_{T_n})$}
\UnaryInfC{$ \emptyset| p ,\Box_i\neg(p\wedge q) \ra$}
\RightLabel{\scriptsize $(R\neg)$}
\UnaryInfC{$ \emptyset| p \ra \neg \Box_i\neg(p\wedge q) $}
\DisplayProof
    
\end{center}

\begin{defn}
    Let $A$ be a formula in $\mathcal{L}^1$,  $\mathsf{b} (A)$ denotes the number of boxed subformulas in  $A$. Given a set $\Gamma$, $\mathsf{b} (\Gamma)$ denotes the sum of all $\mathsf{b} (A)$ for  $A\in \Gamma$. 
\end{defn}
It is noted that  $\mathsf{b} (\Gamma)$ was defined for a set $\Gamma$, not a multiset.


\begin{defn}

Given multi-sets $\Gamma,\Delta,\Gamma^\prime,\Delta^\prime$ of formulas, 
\begin{center}
    $\langle\mathsf{b}(\Gamma),\mathsf{wt}(\Delta)\rangle <
\langle\mathsf{b}(\Gamma^\prime),\mathsf{wt}(\Delta^\prime)\rangle$
\end{center}
denotes a lexicographical order on a pair of natural number. 
We define a well-ordered relation of {\bf T}-sequent.

\begin{center}
$(\Sigma|\Gamma\ra \Delta)\prec (\Sigma^\prime|\Gamma^\prime\ra \Delta^\prime) $ if and only if 

$\langle \mathsf{b}( \Sigma,\Gamma, \Delta),\mathsf{wt}(\Gamma,\Delta)\rangle    < 
\langle \mathsf{b}( \Sigma^\prime,\Gamma^\prime, \Delta^\prime),\mathsf{wt}(\Gamma^\prime,\Delta^\prime)\rangle$     
\end{center}

\end{defn}

\begin{prop}
\label{prop:gktplus terminate}

A backward proof-searching in $\gktnplus$ always terminates.
\end{prop}
\begin{proof}
   Consider a rule $(\circ)$ in $\gktnplus$ as

\begin{center}
    \AxiomC{$\Sigma|\Gamma\ra \Delta$}
        \AxiomC{$\dots$}
    \RightLabel{\scriptsize $(\circ)$}
    \BinaryInfC{$\Sigma^\prime|\Gamma^\prime\ra \Delta^\prime $}
    \DisplayProof
\end{center}
   For any premises $\Sigma|\Gamma\ra \Delta$ and the conclusion $\Sigma^\prime|\Gamma^\prime\ra \Delta^\prime $, $(\Sigma|\Gamma\ra \Delta)\prec (\Sigma^\prime|\Gamma^\prime\ra \Delta^\prime) $. This can be checked from the following arguments.
   
    Let  $(\circ)$ be an arbitrary logical rule.  We have $\mathsf{b}( \Sigma ,\Gamma , \Delta )=\mathsf{b}( \Sigma^\prime,\Gamma^\prime, \Delta^\prime)$, however $\mathsf{wt}(\Gamma,\Delta)<\mathsf{wt}(\Gamma^\prime,\Delta^\prime)$.
   Let  $(\circ)$ be   $(\Box_{Tn}^+) $. We still have   $\mathsf{b}( \Sigma ,\Gamma , \Delta )=\mathsf{b}( \Sigma^\prime,\Gamma^\prime, \Delta^\prime)$, and $\mathsf{wt}(\Gamma,\Delta)<\mathsf{wt}(\Gamma^\prime,\Delta^\prime)$. 
 Let  $(\circ)$ be $(\Box_{Kn}^+) $. We have $\mathsf{b}( \Sigma ,\Gamma , \Delta )<\mathsf{b}( \Sigma^\prime,\Gamma^\prime, \Delta^\prime)$.
\end{proof}

\begin{prop} The following weakening rules are  admissible in $\gktnplus$.

    \begin{center}
\begin{tabular}{lll} 

\AxiomC{$\Sigma|\Gamma \ra \Delta$}
\RightLabel{\scriptsize $(RW)$}
\UnaryInfC{$\Sigma|\Gamma \ra \Delta , A$}
\DisplayProof

   & 

\AxiomC{$\Sigma|\Gamma \ra \Delta$}
\RightLabel{\scriptsize $(LW)$}
\UnaryInfC{$\Sigma|A,\Gamma \ra \Delta$}
\DisplayProof
&
\AxiomC{$\Sigma|\Gamma \ra \Delta$}
\RightLabel{\scriptsize $(LW^+)$}
\UnaryInfC{$\Sigma,\Box_i A|\Gamma \ra \Delta$}
\DisplayProof

\end{tabular}
\end{center}
\end{prop}

    

\begin{prop}
    For any formula $A$, a sequent $\emptyset |A\ra A$ is derivable in  $\gktnplus$.
\end{prop}
\begin{proof}
    When $A$ is in the form of $\Box_i B$, we can have the following derivation.

\begin{center}
\noLine
\AxiomC{I.H.}
\UnaryInfC{$\emptyset| B \ra B$}
\RightLabel{\scriptsize$(\Box_{Kn}^+) $ }
\UnaryInfC{$\Box_i B | \emptyset \ra \Box_i B$ }
\RightLabel{\scriptsize$(LW) $ }
\UnaryInfC{$\Box_i B | B \ra \Box_i B$ }
\RightLabel{\scriptsize$(\Box_{Tn}^+) $ }
\UnaryInfC{$\emptyset| \Box_i B \ra \Box_i B$ }
\DisplayProof
\end{center}
\end{proof}

\begin{prop} 
\label{prop:inverti of t plus rule}
The $(\Box^+_{T_n})$ rule is height-preserving invertible  in $\gktnplus$.
    
\end{prop}

\begin{prop} The following contraction rules are  height-preserving admissible in $\gktnplus$.

  \begin{center}
\begin{tabular}{lll}

            \AxiomC{$\Sigma|\Gamma \ra \Delta , A,A$}
\RightLabel{\scriptsize $(RC)$}
\UnaryInfC{$\Sigma|\Gamma \ra \Delta ,A $}
\DisplayProof

&

\AxiomC{$\Sigma|A,A,\Gamma \ra \Delta$}
\RightLabel{\scriptsize $(LC)$}
\UnaryInfC{$\Sigma|A,\Gamma \ra \Delta$}
\DisplayProof

&
\AxiomC{$\Sigma,\Box_iA,\Box_iA|\Gamma \ra \Delta$}
\RightLabel{\scriptsize $(LC^+)$}
\UnaryInfC{$\Sigma,\Box_iA|\Gamma \ra \Delta$}
\DisplayProof

\end{tabular}

    \end{center}
    
\end{prop}

\begin{prop} The following cut rules are  height-preserving admissible in $\gktnplus$.

\begin{center}
\begin{tabular}{lll}

 \AxiomC{$\emptyset|\Gamma \ra  \Delta, C$}
\AxiomC{$\emptyset|C,\Gamma ^{\prime} \ra \Delta ^{\prime} $}
\RightLabel{\scriptsize $(Cut)$}
\BinaryInfC{$\emptyset|\Gamma,\Gamma^{\prime} \ra \Delta, \Delta  ^{\prime} $}
\DisplayProof

&

 \AxiomC{$\Sigma|\Gamma \ra  \Delta, \Box_i C$}
\AxiomC{$\Box_i C,\Sigma^\prime |,\Gamma ^{\prime} \ra \Delta ^{\prime} $}
\RightLabel{\scriptsize $(Cut^+)$}
\BinaryInfC{$\Sigma^\prime, \Sigma|\Gamma,\Gamma^{\prime} \ra \Delta, \Delta  ^{\prime} $}
\DisplayProof

\end{tabular}
\end{center}
\end{prop}

\begin{proof}
    We need to do a simultaneous induction to show the  above results. 
\end{proof}
 As was mention in \cite{bilkova2007uniform}, $(Cut^+)$ with a formula in a general form cannot enjoy the admissibility. However, this will not affect the proof of UIP.


\begin{prop}
    \label{prop:empty sigma}
        Given multi-sets $\Gamma,\Delta$ of formulas in $\mathcal{L}^1$,
    \begin{center}
         if $\gktnplus\vdash \Sigma|\Gamma\ra \Delta$ then $\gktn\vdash \Sigma,\Gamma\ra \Delta$
    \end{center}
\end{prop}

\begin{lem}
\label{lem:equipplloent of gktn and gktnplus}
    Given multi-sets $\Gamma,\Delta$ of formulas in $\mathcal{L}^1$,
    \begin{center}
        $\gktn\vdash \Gamma\ra \Delta$ if and only if $\gktnplus\vdash \emptyset|\Gamma\ra \Delta$
    \end{center}
\end{lem}
\begin{proof}
    The right-to-left direction can be obtained from the Proposition \ref{prop:empty sigma}.
    The left-to-right direction can be shown by induction on the derivation. 
\end{proof}

\subsection{Main theorem of $\gktn$}
\label{sec:main thm of T}

Then, similar to Definition \ref{dfn:Ap formula in Gkn}, we can define $\mathcal{A}$-formulas in system-{\bf T} as follows.

\begin{defn}
\label{dfn:ap formula for t}
      Let $\Gamma,\Delta$ be finite multi-sets of formulas,  $\Sigma$ be a finite multi-set of out-most boxed formulas,  $p$ be a propositional variable.
      An $\mathcal{A}$-formula $\mathcal{A}_p(\Sigma|\Gamma;\Delta)$  is defined inductively as follows.

\[
\begin{array}
{|c|c|c|c|}\hline & \Sigma|\Gamma;\Delta\text{ matches} & \mathcal{A}_p(\Sigma|\Gamma;\Delta)\text{ equals} \\
\hline

~~1~~ & \Sigma|\Gamma',p;\Delta',p & \top \\

~~2~~ & \Sigma|\Gamma',\bot;\Delta & \top  \\
~~3~~ & \Sigma|\Gamma',C_1\wedge C_2;\Delta & \mathcal{A}_p(\Sigma|\Gamma',C_1,C_2;\Delta) \\
4 & \Sigma|\Gamma;C_1\wedge C_2,\Delta' & \mathcal{A}_p(\Sigma|\Gamma;C_1,\Delta')\land \mathcal{A}_p(\Sigma|\Gamma;C_2,\Delta') \\

5 & \Sigma|\Gamma',C_1\vee C_2;\Delta & \mathcal{A}_p(\Sigma|\Gamma',C_1;\Delta)\land \mathcal{A}_p(\Sigma|\Gamma',C_2;\Delta) \\

~~6~~ & \Sigma|\Gamma;C_1\vee C_2,\Delta' & \mathcal{A}_p(\Sigma|\Gamma';C_1,C_2,\Delta') \\

~~7~~ & \Sigma|\Gamma',\neg C;\Delta & \mathcal{A}_p(\Sigma|\Gamma';C,\Delta) \\

8& \Sigma|\Gamma;\neg C,\Delta' & \mathcal{A}_p(\Sigma|\Gamma,C;\Delta') \\

9 & \Sigma|\Gamma',C_1\rightarrow C_2;\Delta & \mathcal{A}_p(\Sigma|\Gamma';\Delta,C_1)\land \mathcal{A}_p(\Sigma|\Gamma',C_2;\Delta) \\
10 & \Sigma|\Gamma;C_1\rightarrow C_2,\Delta' & \mathcal{A}_p(\Sigma|\Gamma,C_1; \Delta',C_2) \\

11 & \Sigma|\Gamma',\Box_i C;\Delta & \mathcal{A}_p(\Sigma,\Box_i C|\Gamma',C;\Delta) \\

12\dagger& ~~ \overrightarrow{\Box_{s_m}\sigma_m } |\Phi; \overrightarrow{\Box_{d_n}\delta_n }, \Psi ~~  &  \mathsf{X} \\ 
 
\hline
\end{array}
\]
{\footnotesize $\dagger:$ $\Phi$ and  $\Psi$ are multisets containing only propositional variables or $\bot$, 
besides $q,r$ differ from $p$. 
Furthermore $\Phi\cup \overrightarrow{\Box_{s_m}\sigma_m }\cup\overrightarrow{\Box_{d_n}\delta_n }\cup \Psi$ is not empty.}

The formula $\mathsf{X}$ is:

\[\bigvee_{q\in\Psi}q
\vee
\bigvee_{r\in\Phi}\neg r 
\vee 
\bigvee_{\Box_{s_j}\sigma_j \in \{\overrightarrow{\Box_{s_m}\sigma_m }\}} \Diamond_{s_j} \mathcal{A}_p (   \emptyset| \{  \overrightarrow{\Box_{s_m\sigma_m}} \}^{\flat_{s_j}} ;\emptyset) 
\]

\[
  \vee \bigvee_{\Box d_i\delta_i \in \{\overrightarrow{\Box_{d_n}\delta_n } \}}   \Box{d_i} \mathcal{A}_p (  \emptyset| \{\overrightarrow{\Box_{s_m}\sigma_m }  \}^{\flat_{di}}; \delta_i   )
\]


Recall that for any formula $A$, $n\in\mathbb{N}$, $\overrightarrow{A_n}$ stands for  $A_1,\cdots, A_n$. 
$\Gamma^{\flat_i}=\{A|\Box_iA\in \Gamma\}$,

The formula $\mathcal{A}_p(\Sigma|\Gamma;\Delta)$ is defined in the following procedure: at first, the lines $1-11$ are repeatedly applied until it reaches  the line $12$.  We repeat this procedure until $\Sigma|\Gamma;\Delta$ cannot  match any lines in the table, then $\mathcal{A}_p(\Sigma|\Gamma;\Delta)$ is defined as $\bot$. Especially $\mathcal{A}_p(\Sigma|\emptyset;\emptyset)$ is defined as $\bot$.

We can define a well-order relation of $\mathcal{A}$-formulas as follows:
\begin{center}
$\mathcal{A}_p (\Sigma|\Gamma; \Delta)\prec \mathcal{A}_p (\Sigma^\prime|\Gamma^\prime; \Delta^\prime) $ if and only if

$(\Sigma|\Gamma\ra \Delta)\prec (\Sigma^\prime|\Gamma^\prime\ra \Delta^\prime) $ 
\end{center}
Given the fact that all back proof-search in $\gktnplus$ always terminates (in Proposition \ref{prop:gktplus terminate}), we can see that 
such a formula can always be determined.

\end{defn}

\begin{prop}
\label{prop:initial casein main thm in gktnplus}
   Let $\Gamma,\Delta$ be finite multi-sets of formulas,  $\Sigma$ be a finite multi-set of out-most boxed formulas,  $p,q$ be  propositional variable, such that $p\neq q$.  

\begin{enumerate}
    \item $\gktnplus\vdash \emptyset|q\ra \mathcal{A}_p(\Sigma|\Gamma;\Delta,q)$;
    \item $\gktnplus\vdash \emptyset|\neg q\ra \mathcal{A}_p(\Sigma|\Gamma,q;\Delta)$;
    \item $\gktnplus\vdash  \emptyset|\emptyset \ra \mathcal{A}_p(\Sigma|q,\Gamma;\Delta,q)$.
\end{enumerate}\end{prop}

\begin{thm}
\label{thm:main theorem of gktnplus}
   Let $\Gamma,\Delta$ be finite multi-sets of formulas,  $\Sigma$ be a finite multi-set of out-most boxed formulas.
    For every propositional variable $p$ there exists a formula $\mathcal{A}_p(\Sigma|\Gamma;\Delta)$ such that:
    \begin{enumerate}[(i)]
        \item  
            $\mathsf{V}( \mathcal{A}_p(\Sigma|\Gamma,\Delta)) \subseteq \mathsf{V}(\Sigma\cup\Gamma\cup \Delta)\backslash\{p\}    $
        \item $\gktnplus\vdash    \Sigma|\Gamma , \mathcal{A}_p(\Sigma|\Gamma; \Delta) \ra \Delta$
        \item given finite multi-sets $\Pi, \Lambda$ of formulas, $\Theta$  of out-most boxed formulas, such that 
        \begin{center}
            $p\notin \mathsf{V}(\Pi\cup\Lambda\cup\Theta)$ and $\gktnplus \vdash \Theta,\Sigma|  \Pi,\Gamma \ra \Delta, \Lambda$
        \end{center}
        then 
        \begin{center}
            $\gktnplus\vdash  \emptyset|\Theta,\Pi \ra  \mathcal{A}_p(\Sigma|\Gamma; \Delta) ,\Lambda$
        \end{center}
    \end{enumerate}
    
\end{thm}

\begin{proof}

We proceed in a similar way to the proof of Theorem \ref{thm:main theorem of gkn}.

The proof of (i) can be obtained by inspecting the table in Definition \ref{dfn:ap formula for t}.

The proof of  (ii) can be proved by induction on the weight of $\mathcal{A}_p(\Gamma,\Delta)$. We prove $\gktnplus\vdash    \Gamma , \mathcal{A}_p(\Gamma; \Delta) \ra \Delta$ for each line of the table in  Definition \ref{dfn:Ap formula in Gkn}.
The cases of lines from 1 to 11 are easy.
We only concentrate on the case of  line 12.

The proof is similar to Theorem \ref{thm:main theorem of gkn}, only significant cases are shown here.

\begin{itemize}
    \item For each  $\Box_{s_j}\sigma_j \in \{\overrightarrow{\Box_{s_m}\sigma_m }\}$,

$ \gktnplus\vdash    \emptyset|\{ \overrightarrow{\Box_{s_m\sigma_m}} \}^{\flat_{s_j}},\mathcal{A}_p (   \emptyset| \{ \overrightarrow{\Box_{s_m\sigma_m}} \}^{\flat_{s_j}} ;\emptyset)  \ra    $  from induction hypothesis. After applying $(R\neg)$,  $(\Box^{+}_{Kn})$,$(L\neg)$, and weakening rules we obtain

$ \gktnplus \vdash \overrightarrow{\{\Box_{s_m}\sigma_m \}}|\Phi, \Diamond_{s_j} \mathcal{A}_p (   \emptyset| \{ \overrightarrow{\Box_{s_m\sigma_m}} \}^{\flat_{s_j}} ;\emptyset)  \ra  \overrightarrow{\Box_{d_n}\delta_n },\Psi   $ 
\item  
For each $\Box d_i\delta_i \in \{\overrightarrow{\Box_{d_n}\delta_n } \}$, by induction hypothesis, we obtain:
 
$  \gktnplus \vdash \emptyset| \{\overrightarrow{\Box_{s_m}\sigma_m }  \}^{\flat_{di}}, \mathcal{A}_p ( \emptyset|  \{\overrightarrow{\Box_{s_m}\sigma_m }  \}^{\flat_{di}}; \delta_i   )  
\ra \delta_i  .$

 After applying $ (\Box^{+}_{Kn})$, $ (\Box^{+}_{Tn})$  and weakening rules we obtain the desired result.

$ \gktnplus \vdash \overrightarrow{\{\Box_{s_m}\sigma_m \}}|\Phi, \Box_{d_i} \mathcal{A}_p (   \emptyset| \{ \overrightarrow{\Box_{s_m\sigma_m}} \}^{\flat_{s_j}} ;\delta_i)  \ra  \overrightarrow{\Box_{d_n}\delta_n },\Psi   $ 

\end{itemize}
After apply $(L\lor)$ for finite many times, we can obtain the desired result.

Next, in the proof of (iii), we consider the last rules applied in the derivation. We focus on the following cases. 

When it is an initial sequent, we need the Proposition \ref{prop:initial casein main thm in gktnplus}.

When the last rule is $(\Box^{+}_{T_n})$, we need to consider the cases:
\begin{enumerate}
    \item When principal formula $\Box_i A $ appears in $\Pi$, we need to apply the invertibility in Proposition \ref{prop:inverti of t plus rule} and contraction rules.
    \item When principal formula $\Box_i A $ appears in $\Gamma$, the proof can be obtained directly from the definition.
    
\end{enumerate}
 
When the last rule is $(\Box^{+}_{K_n})$, 
the arguments are divided into the following cases:
\begin{enumerate}
    \item The right principal formula $\Box_i A$ is in the multiset $\Lambda$.

    \begin{center}

\AxiomC{$\emptyset|\Theta'',\Sigma''\ra A$}
\RightLabel{$(\Box_{K_n^+}) $}
\UnaryInfC{$\Theta',\Box_i\Theta'',\overrightarrow{\Box s_m \sigma_m},\Box_i\Sigma''|\Pi, \Gamma\ra \Box_i A, \Lambda',  \overrightarrow{\Box d_n \delta_n}, \Psi $}
\DisplayProof

\end{center}

 In this case, $\Theta$ is $\Theta',\Box_i\Theta''$, and $\Sigma$ is $\overrightarrow{\Box s_m \sigma_m},\Box_i\Sigma''$   ($\Box_i $ is not among $\overrightarrow{\Box s_m }$) , and $\Lambda$ is  $\Box_i A, \Lambda'$, and $\Delta$ is $\overrightarrow{\Box d_n \delta_n}, \Psi$.

From assumption, we find that $p\notin\mathsf{V}(\Theta'' \cup\{A\})$. Then, according to induction hypothesis, we can derive: $\emptyset|\Theta'' \ra \mathcal{A}_p(\emptyset| \{\overrightarrow{\Box s_m \sigma_m},\Box_i\Sigma''  \}^{\flat_i} ;\emptyset) ,A$.
Then after applying $(L\neg)$ and $(\Box^+_{K_n})$, we derive:
\begin{center}
    $\Theta',\Box_i\Theta'', \Box_i  \neg\mathcal{A}_p(\emptyset| \{\overrightarrow{\Box s_m \sigma_m},\Box_i\Sigma''  \}^{\flat_i} ;\emptyset) |\Pi \ra \Box_iA,\Lambda'$.
\end{center}
Next, after applying weakening, $(\Box^+_{T_n})$ and $(R\neg)$, we derive:
\begin{center}
    $\emptyset|\Theta',\Box_i\Theta'', \Pi \ra \Diamond_i  \mathcal{A}_p(\emptyset| \{\overrightarrow{\Box s_m \sigma_m},\Box_i\Sigma''  \}^{\flat_i} ;\emptyset),\Box_iA,\Lambda'$.
\end{center}
Finally we can obtain the results by applying weakening and $(R\lor)$ for finite many times.

        \item The right principal formula $\Box_i A$ is in the multiset $\Delta$.

\begin{center}

\AxiomC{$\emptyset|\Theta'',\Sigma''\ra A$}
\RightLabel{$(\Box_{K_n^+}) $}
\UnaryInfC{$\Theta',\Box_i\Theta'',\overrightarrow{\Box s_m \sigma_m},\Box_i\Sigma''|\Pi, \Gamma\ra \Box_i A,  \overrightarrow{\Box d_n \delta_n}, \Psi,\Lambda $}
\DisplayProof
\end{center}

 In this case, $\Theta$ is $\Theta',\Box_i\Theta''$, and $\Sigma$ is $\overrightarrow{\Box s_m \sigma_m},\Box_i\Sigma''$,  ($\Box_i $ is not among $\overrightarrow{\Box s_m }$) and $\Delta$ is $\Box_i A,\overrightarrow{\Box d_n \delta_n}, \Psi$.

From assumption, we find that $p\notin\mathsf{V}(\Theta'' )$. Then, according to induction hypothesis, we can derive: $\emptyset|\Theta'' \ra \mathcal{A}_p(\emptyset| \{\overrightarrow{\Box s_m \sigma_m},\Box_i\Sigma''  \}^{\flat_i} ;A)$.
After applying $(\Box^+_{K_n})$, we derive:
\begin{center}
    $\Theta',\Box_i\Theta''  |\Pi \ra \Box_i  \mathcal{A}_p(\emptyset| \{\overrightarrow{\Box s_m \sigma_m},\Box_i\Sigma''  \}^{\flat_i} ;A),\Lambda$.
\end{center}
Next, after applying weakening and $(\Box^+_{T_n})$ for finite many times we derive:
\begin{center}
    $\emptyset|\Theta',\Box_i\Theta''  ,\Pi \ra \Box_i  \mathcal{A}_p(\emptyset| \{\overrightarrow{\Box s_m \sigma_m},\Box_i\Sigma''  \}^{\flat_i} ;A),\Lambda$.
\end{center}
Finally we can obtain the results by applying weakening and $(R\lor)$ for finite many times. \qedhere

\end{enumerate}
\end{proof}

Then, we can transfer the above results of $\gktnplus$ to $\gktn$.

\begin{cor}
\label{cor: cor main theorem of gktn}

    Let $\Gamma,\Delta$ be finite multi-sets of formulas. For every propositional variable $p$ there exists a formula $\mathcal{A}_p(\Gamma;\Delta)$ such that:
    \begin{enumerate}[(i)]
        \item  
            $\mathsf{V}( \mathcal{A}_p(\Gamma;\Delta)) \subseteq \mathsf{V}(\Gamma\cup \Delta)\backslash\{p\}    $
        \item $\gktn\vdash    \Gamma , \mathcal{A}_p(\Gamma; \Delta) \ra \Delta$
        \item given finite multi-sets $\Pi, \Lambda$ of formulas such that 
        \begin{center}
            $p\notin \mathsf{V}(\Pi\cup\Lambda)$ and $\gktn\vdash  \Pi,\Gamma \ra \Delta, \Lambda$
        \end{center}
        then 
        \begin{center}
            $\gktn\vdash  \Pi \ra  \mathcal{A}_p(\Gamma; \Delta) ,\Lambda$
        \end{center}
    \end{enumerate}
    
\end{cor}
\begin{proof}
    Let $\mathcal{A}_p(\Gamma;\Delta) $ be $\mathcal{A}_p(\emptyset|\Gamma;\Delta) $. Then, we apply Lemma \ref{lem:equipplloent of gktn and gktnplus} and Theorem \ref{thm:main theorem of gktnplus}
. \end{proof}

\begin{cor}
 \label{cor:uip in ktn first}
    Uniform interpolation property is satisfied in $\gktn$.
    
\end{cor}
\begin{proof}
    We proceed in a similar way to the proof of Corollary \ref{cor:uip in kn, kdn, first}. 
\end{proof}

 \section{Propositional quantifiers and UIP}
\label{sec:2nd order}


In this section, we show that quantification over propositional variables can be modeled by UIP in  $\gkn$, $\gkdn$ and $\gktn$.
This also provide an alternative method to show the uniform interpolation property from main theorem as in \cite{pitts1992interpretation,bilkova2007uniform}.

\begin{table}[htb]
\caption{Sequent Calculi $\gksecn$, $\gkdsecn$ and $\gktsecn$.}
\label{table:gksecn}
\begin{center}
\hrule
\begin{tabular}{ll}

\multicolumn{2}{c}{ Sequent Calculus $\gksecn$, $\gkdsecn$ and $\gktsecn$
 } \\
 
\multicolumn{2}{c}{ Adding the following rules  to $\gkn$ ,$\gkdn$  and $\gktn$  respectively
 } \\

{\bf Initial Sequent} 

&

$\forall p \Box_i A \ra \Box_i \forall p A$

\\

{ \bf Structural Rules}   & 
\AxiomC{$\Gamma \ra \Delta$}
\RightLabel{\scriptsize $(LW)$}
\UnaryInfC{$A,\Gamma \ra \Delta$}
\DisplayProof

\AxiomC{$\Gamma \ra \Delta$}
\RightLabel{\scriptsize $(RW)$}
\UnaryInfC{$\Gamma \ra \Delta , C$}
\DisplayProof
\\

&

\AxiomC{$\Gamma \ra \Delta , A,A$}
\RightLabel{\scriptsize $(RC)$}
\UnaryInfC{$\Gamma \ra \Delta ,A $}
\DisplayProof

\AxiomC{$A,A,\Gamma \ra \Delta$}
\RightLabel{\scriptsize $(LC)$}
\UnaryInfC{$A,\Gamma \ra \Delta$}
\DisplayProof
\\

\, &    
\AxiomC{$\Gamma \ra  \Delta, C$}
\AxiomC{$C,\Gamma ^{\prime} \ra \Delta ^{\prime} $}
\RightLabel{\scriptsize $(Cut)$}
\BinaryInfC{$\Gamma,\Gamma^{\prime} \ra \Delta, \Delta  ^{\prime} $}
\DisplayProof
\\

  {\bf Prop Quantifier rule}

& 

\AxiomC{$ \Gamma , A(p/B)\ra \Delta   $ }
\RightLabel{\scriptsize $( L\forall )$ }
\UnaryInfC{$ \Gamma ,\forall p A \ra \Delta$}
\DisplayProof

\AxiomC{$ \Gamma \ra A,\Delta   $ }
\RightLabel{\scriptsize $( R\forall )$$\ddagger$ }
\UnaryInfC{$ \Gamma  \ra\forall p A, \Delta$}
\DisplayProof
\\
\\
\multicolumn{2}{l}{ {\footnotesize $\ddagger$:  $p$ is not free in $\Gamma$,$\Delta$.  }
   }

\\

\hline

\end{tabular}
\end{center}
\end{table}

The rules in the table \ref{table:gksecn} are originated from \cite{bilkova2007uniform}.
The propositional version of the Barcon formula is represented by the initial sequent   $\forall p \Box_i A \ra \Box_i \forall p A$ \, since we want to have constant domain for every possible world in this logic. However, this will make it difficult to prove cut-elimination theorem. (cf. \cite{bull1969modal,bilkova2007uniform})

\begin{prop}[Substitution]
    Whenever a sequent $\Gamma\ra \Delta $ is derivable in $\gksecn$, $\gkdsecn$ and $\gktsecn$, a sequent in the form of
    $\Gamma [p/B]\ra \Delta [p/B]$ is also derivable.
\end{prop}

We define a translation from $\mathcal{L}^1$ to $\mathcal{L}^2$. 
Let $\mathbf{L}\in \{ \mathbf{K}_n, \mathbf{KD}_n,\mathbf{KT}_n\}$. It can be shown that for an arbitrary   $\mathcal{L}^2$ formula  derivable in $\mathsf{G}(\mathbf{L}^2)$ , its translated formula in $\mathcal{L}^1$  is also   derivable in $\mathsf{G}(\mathbf{L})$.

    Recall that,
    for a formula $B\in \mathcal{L}^1$, a propositional variable $p$,  $\mathcal{A}_p (B)$  is $\mathcal{A}_p (\emptyset;B)$ .

\begin{defn}
    Given a formula $A$ in $\mathcal{L}^2$, we define $A^\ast$ in $\mathcal{L}^2$ inductively as follows:
    \begin{itemize}
        \item $p^\ast := p$
        
        \item $(B\circ C)^\ast := B^\ast \circ C^\ast$ ($\circ\in \{\wedge,\lor, \rightarrow \}$)

                \item $(\neg B)^\ast := \neg B ^\ast$
                \item $(\Box_i B)^\ast := \Box_i B^\ast$
                        \item $(\forall p B)^\ast := \mathcal{A} p B^\ast$
        
    \end{itemize}
    For a set $\Gamma$ of formulas, $\Gamma^\ast = \{ A^\ast| A\in \Gamma \}.$
\end{defn}

We can observe that in a quantifier-free formula $B$, $B^\ast = B$.

As a corollary of our main theorem \ref{thm:main theorem of gkn} and corollary \ref{cor: cor main theorem of gktn}, we can obtain the following results to prove the preservation of quantifier rules.
\begin{cor}
\label{cor:deri}
   Let $\mathbf{L}\in \{ \mathbf{K}_n, \mathbf{KD}_n,\mathbf{KT}_n\}$. In the language $\mathcal{L}^1$, let $C$ be a formula, $\Gamma,\Delta$ be multisets of formulas not containing $p$. 
\begin{enumerate}

    \item  $\mathsf{G}(\mathbf{L})\vdash  \mathcal{A}_p (C) \ra  C[p/B] $ for any formula $B$;
    \item   $\mathsf{G}(\mathbf{L})\vdash \Gamma\ra C,\Delta $ implies $\mathsf{G}(\mathbf{L})\vdash \Gamma\ra \mathcal{A}_p (C),\Delta $.
\end{enumerate}
  \end{cor}

\begin{proof}
 Proved
 immediately from (ii) and (iii) of theorem \ref{thm:main theorem of gkn}, corollary \ref{cor: cor main theorem of gktn}.
\end{proof}

\begin{prop}
\label{prop:congurence}
   Let $\mathbf{L}\in \{ \mathbf{K}_n, \mathbf{KD}_n,\mathbf{KT}_n\}$.  $\mathsf{G}(\mathbf{L})\vdash A \leftrightarrow B, C[q/B] \ra C[q/A] $
    
\end{prop}

\begin{cor}
\label{cor:sub}
   Let $\mathbf{L}\in \{ \mathbf{K}_n, \mathbf{KD}_n,\mathbf{KT}_n\}$. Given a formula $C$ and a formula  $B$,  $B$ doesn’t contain $p,q$ and $p\ne q$.
\begin{enumerate}
    \item   $\mathsf{G}(\mathbf{L})\vdash \mathcal{A}p(C[q/B]) \ra (\mathcal{A}p(C))[q/B]$ 
    \item    $\mathsf{G}(\mathbf{L})\vdash (\mathcal{A}p(C))[q/B] \ra   \mathcal{A}p(C[q/B])$.
\end{enumerate}
\end{cor}

\begin{proof}    To see 1, we consider the following derivation.
    \begin{center}
    \noLine
    \AxiomC{\footnotesize (ii) of Theorem \ref{thm:main theorem of gkn},  Corollary \ref{cor: cor main theorem of gktn}}
    \UnaryInfC{$\mathcal{A}_p (C [q/B])\ra C[q/B] $}
    \noLine
    \AxiomC{Proposition \ref{prop:congurence} }
    \UnaryInfC{$q\leftrightarrow B, C[q/B]\ra C[q/q]$}
    \RightLabel{$(Cut)$}
    \BinaryInfC{$ q\leftrightarrow B, \mathcal{A}_p ( C[q/B]) \ra C$}
    \DisplayProof
    \end{center}
From (iii) of theorem \ref{thm:main theorem of gkn}, corollary \ref{cor: cor main theorem of gktn} and $p$ not occurring in $ q\leftrightarrow B, \mathcal{A}_p ( C[q/B])$,
we obtain $\mathsf{G}(\mathbf{L})\vdash q\leftrightarrow B, \mathcal{A}_p ( C[q/B]) \ra \mathcal{A}_p(C)$.
Then, applying substitution of $[q/B] $ and $(cut)$ with $\ra B \leftrightarrow B$,
    we obtain $\mathsf{G}(\mathbf{L})\vdash \mathcal{A}p(C[q/B]) \ra (\mathcal{A}p(C))[q/B]$.

    To see 2, at first we have $\mathsf{G}(\mathbf{L})\vdash  \mathcal{A}_p (C) \ra C $ from (ii) of theorem \ref{thm:main theorem of gkn}, corollary \ref{cor: cor main theorem of gktn}. By substitution we have $\mathsf{G}(\mathbf{L})\vdash  \mathcal{A}_p (C)[q/B]  \ra C [q/B]$.
     Since the antecedent doesn’t contain $p$,  $\mathsf{G}(\mathbf{L})\vdash (\mathcal{A}p(C))[q/B] \ra   \mathcal{A}p(C[q/B])$ from (iii) of theorem \ref{thm:main theorem of gkn}, corollary \ref{cor: cor main theorem of gktn}.
\end{proof}

We can immediately infer the following results.

\begin{prop}
       Let $\mathbf{L}\in \{ \mathbf{K}_n, \mathbf{KD}_n,\mathbf{KT}_n\}$. Given a formula $C$ and a formula  $B$,
       $B$ doesn’t contain $p,q$ and $p\ne q$.
\begin{enumerate}
    \item   $\mathsf{G}(\mathbf{L})\vdash (C[q/B])^\ast \ra C^\ast[q/B]$;
    \item    $\mathsf{G}(\mathbf{L})\vdash C^\ast[q/B] \ra   (C[q/B])^\ast$.
\end{enumerate}
\end{prop}

\begin{cor}
\label{cor:main cor}
      Let $\mathbf{L}\in \{ \mathbf{K}_n, \mathbf{KD}_n,\mathbf{KT}_n\}$. For multi-sets $\Gamma,\Delta$ of formulas in $\mathcal{L}^2$, if $\mathsf{G}(\mathbf{L}^2)\vdash \Gamma\ra \Delta$ then $\mathsf{G}(\mathbf{L})\vdash \Gamma^\ast \ra \Delta^\ast$.
\end{cor}
\begin{proof}
      By induction on the derivation and applying Corollary \ref{cor:deri} and Corollary \ref{cor:sub}.
      It is noted that, if $p$ is not free in a formula $A\in \mathcal{L}^2$, then $p$ does not occur in $ A^\ast\in \mathcal{L}^1$. 
      
      In the case of new initial sequent $\forall p \Box_i B \ra \Box_i \forall p B$, 
       from the definition of $\mathcal{A}_p$ formulas in definition \ref{dfn:Ap formula in Gkn}, $\mathcal{A}_p(\emptyset; \Box_i B )$ equals to  $\Box_i\mathcal{A}_p(\emptyset;  B )$.
     Then we obtain 
     $\mathsf{G}(\mathbf{L})\vdash \mathcal{A}_p(\Box_i B)\ra \Box_i \mathcal{A}_p (B)$.
\end{proof}

\begin{lem}
\label{lem:uip in second}
    Uniform interpolation properties are satisfied in $\gksecn$, $\gkdsecn$, $\gktsecn$ in language $\mathcal{L}^2$.
\end{lem}


  Let $\mathbf{L}\in \{ \mathbf{K}_n, \mathbf{KD}_n,\mathbf{KT}_n\}$.
For any formula $A(\overrightarrow{p},\overrightarrow{q})$, any  propositional  variables $\overrightarrow{q}$, such that all $q$ are different from all $p$,
 there exists a formula (post-interpolant) $\mathcal{I}_{post}(A, \overrightarrow{q} )$  such that:
\begin{enumerate}
    \item $A(\overrightarrow{p},\overrightarrow{q})\ra\mathcal{I}_{post} (A, \overrightarrow{q} )  $ is derivable in $\mathsf{G}(\mathbf{L}^2)$;
    \item  for any formula $B(\overrightarrow{q},\overrightarrow{r})$, where
    all $r$ are different from all $q$, 
    if $A(\overrightarrow{p},\overrightarrow{q})  \ra  B(\overrightarrow{q},\overrightarrow{r})$ is derivable in $\mathsf{G}
    (\mathbf{L}^2)$ then   $\mathcal{I}_{post}(A,\overrightarrow{q} )\ra B(\overrightarrow{q},\overrightarrow{r}) $ is derivable in $\mathsf{G}
    (\mathbf{L}^2)$.
\end{enumerate}

Furthermore, for any formula $B(\overrightarrow{q},\overrightarrow{r})$, any  propositional  variables $\overrightarrow{q}$, such that all $q$ are different from all $r$,
 there exists a formula (pre-interpolant) $\mathcal{I}_{pre}(B, \overrightarrow{q} )$  such that:
\begin{enumerate}
    \item $\mathcal{I}_{pre} (B, \overrightarrow{q} ) \ra B(\overrightarrow{q},\overrightarrow{r})  $ is derivable in $\mathsf{G}(\mathbf{L}^2)$;
    \item  for any formula $A(\overrightarrow{p},\overrightarrow{q})$, where
    all $p$ are different from all $q$, 
    if $A(\overrightarrow{p},\overrightarrow{q})  \ra  B(\overrightarrow{q},\overrightarrow{r})$ is derivable in $\mathsf{G}
    (\mathbf{L}^2)$ then   $
    A(\overrightarrow{p},\overrightarrow{q})\ra
\mathcal{I}_{post}(B,\overrightarrow{q} )
    $ is derivable in $\mathsf{G}(\mathbf{L}^2)$.
\end{enumerate}
\begin{proof}

    Take $\overrightarrow{ \exists p}A (\overrightarrow{p},\overrightarrow{q} )$ as post-interpolant $\mathcal{I}_{post}(A, \overrightarrow{q} )$, 
    $\overrightarrow{ \forall r}B (\overrightarrow{q},\overrightarrow{r} )$ as pre-interpolant $\mathcal{I}_{pre}(B, \overrightarrow{q} )$.

    The case of pre-interpolant can be proved as follows:
    At first,     We can easily prove that $\mathsf{G}(\mathbf{L}^2)\vdash  \overrightarrow{ \forall r}B (\overrightarrow{q},\overrightarrow{r} )\ra B(\overrightarrow{q},\overrightarrow{r})$   from   applying the rule $(L\forall)$ for finite times ;
    Next,  
    $\mathsf{G}(\mathbf{L}^2)\vdash   A(\overrightarrow{p},\overrightarrow{q}) \ra     \overrightarrow{ \forall r}(\overrightarrow{q},\overrightarrow{r} ) $ can be derived from   applying the rule $(R\forall)$ for finite times.
    Next, the case of post-interpolant is a dual. 
\end{proof}

The translation  provides an alternative method to show the uniform interpolation property (already proved in as in Corollary 
\ref{cor:uip in kn, kdn, first} and \ref{cor:uip in ktn first}) from main theorem. 

\begin{cor}
    Uniform interpolation properties are satisfied in $\gkn$, $\gkdn$ and $\gktn$ language $\mathcal{L}^1$.
\end{cor}

\begin{proof}
    Proved from Corollary \ref{cor:main cor} and Lemma \ref{lem:uip in second}.
\end{proof}

 \section{Conclustion and Future Direction}
This paper expands the single-modal systems from  B{\'\i}lkov{\'a}\cite{bilkova2007uniform} to multi-modal systems. 
We have given pure syntactic proof of UIP in multi-modal logic $\mathbf{K_n}$ (in Corollary \ref{cor:uip in kn, kdn, first}) , $\mathbf{KD_n}$ (in Corollary \ref{cor:uip in kn, kdn, first}, it is noted that $\mathbf{KD}$ is not in B{\'\i}lkov{\'a}\cite{bilkova2007uniform}) and $\mathbf{KT_n}$ (in Corollary \ref{cor:uip in ktn first}).
We also show that quantification over propositional variables can be modeled by UIP in these systems.

In the next step, it could be interesting to prove UIP in the intuitionistic multi-modal system. To do this, we may need to make use of the $\mathbf{G4}$-style sequent calculus in Pitts \cite{pitts1992interpretation}.
Another direction is to show the UIP for distributed knowledge. We did not take agents into account in interpolation. However, counting agents in Craig interpolation is very straightforward (for example in \cite{murai2024intui_pal_dis,Su2021_diel_lori}) in epistemic logic with distributed knowledge. It could be interesting to develop a method to count agents.



\bibliographystyle{splncs04}
\bibliography{UIP202505}

\begin{thebibliography}{10}
\providecommand{\url}[1]{\texttt{#1}}
\providecommand{\urlprefix}{URL }
\providecommand{\doi}[1]{https://doi.org/#1}

\bibitem{alassaf_uip_2022}
Alassaf, R., Schmidt, R., Sattler, U.: Saturation-based uniform interpolation for multi-modal logics. In: Fernandez-Duque, D., Palmigiano, A., Pinchinat, S. (eds.) Advances in Modal Logic. pp. 37--57. Advances in Modal Logic, College Publications, United Kingdom (Jul 2022), publisher Copyright: {\textcopyright} 2022 College Publications. All rights reserved.; 14th Conference on Advances in Modal Logic, AiML 2022 ; Conference date: 22-08-2022 Through 25-08-2022

\bibitem{bilkova2007uniform}
B{\'\i}lkov{\'a}, M.: Uniform interpolation and propositional quantifiers in modal logics. Stud Logica  \textbf{85},  1--31 (2007)

\bibitem{bull1969modal}
Bull, R.: On modal logic with propositional quantifiers. The journal of symbolic logic  \textbf{34}(2),  257--263 (1969)

\bibitem{chang1990model}
Chang, C.C., Keisler, H.J.: Model theory. Dover, 3 edn. (1990)

\bibitem{Craig1957}
Craig, W.: Three uses of the herbrand-gentzen theorem in relating model theory and proof theory. The Journal of Symbolic Logic  \textbf{22}(3),  269--285 (1957), \url{http://www.jstor.org/stable/2963594}

\bibitem{hans2009}
van Ditmarsch, H., Herzig, A., Lang, J., Marquis, P.: Introspective forgetting. Synthese  \textbf{169}(2),  405--423 (2009), \url{http://www.jstor.org/stable/40271285}

\bibitem{FANG201951}
Fang, L., Liu, Y., {van Ditmarsch}, H.: Forgetting in multi-agent modal logics. Artificial Intelligence  \textbf{266},  51--80 (2019). \doi{10.1016/j.artint.2018.08.003}

\bibitem{ghilardi1995algebraic}
Ghilardi, S.: An algebraic theory of normal forms. Annals of Pure and Applied Logic  \textbf{71}(3),  189--245 (1995)

\bibitem{ghilardi1995undefinability}
Ghilardi, S., Zawadowski, M.: Undefinability of propositional quantifiers in the modal system s4. Studia Logica  \textbf{55}(2),  259--271 (1995)

\bibitem{Giessen2024_uip_proof}
van~der Giessen, I., Jalali, R., Kuznets, R.: Uniform interpolation via nested sequents and hypersequents. Journal of Logic and Computation p. exae053 (12 2024). \doi{10.1093/logcom/exae053}

\bibitem{hakli2012does}
Hakli, R., Negri, S.: Does the deduction theorem fail for modal logic? Synthese  \textbf{187}(3),  849--867 (2012)

\bibitem{Iemhoff2019_uip}
Iemhoff, R.: Uniform interpolation and sequent calculi in modal logic. Arch. Math. Logic  \textbf{58},  155--181 (2019). \doi{10.1007/s00153-018-0629-0}

\bibitem{Kashima2009}
Kashima, R.: Mathematical logic. Asakura Publishing Co., Ltd (in Japanese) (2009)

\bibitem{Lin1994}
Lin, F., Reiter, R.: Forget it! In: Proceedings of AAAI Fall Symposium on Relevance. pp. 154--159 (1997)

\bibitem{Maehara1961_interpolation}
Maehara, S.: On {C}raig’s interpolation theorem (in {J}apanese). Sugaku  \textbf{12}(4),  235--237 (1961)

\bibitem{murai2024intui_pal_dis}
Murai, R., Sano, K.: Intuitionistic public announcement logic with distributed knowledge. Studia Logica  \textbf{112}(3),  661--691 (2024)

\bibitem{Negri2001}
Negri, S., von Plato: Structural Proof Theory. Cambridge University Press (2001)

\bibitem{Ono1998}
Ono, H.: Proof-theoretic methods in nonclassical logic–an introduction. Theories of types and proofs. Vol.2:207-254  (1998)

\bibitem{pitts1992interpretation}
Pitts, A.M.: On an interpretation of second order quantification in first order intuitionistic propositional logic. The Journal of Symbolic Logic  \textbf{57}(1),  33--52 (1992)

\bibitem{Su2021_diel_lori}
Su, Y., Murai, R., Sano, K.: On {A}rtemov and {P}rotopopescu's intuitionistic epistemic logic expanded with distributed knowledge. In: Ghosh, S., Icard, T. (eds.) Logic, Rationality, and Interaction: 8th International Workshop, Lori 2021, Xi'an, China, October 16-18, 2021, Proceedings, pp. 216--231. Springer Verlag (2021)

\bibitem{troelstra2000basic}
Troelstra, A.S., Schwichtenberg, H.: Basic proof theory. No.~43, Cambridge University Press (2000)

\bibitem{visser1996bisimulations}
Visser, A.: Bisimulations, model descriptions and propositional quantifiers. Logic group preprint series  \textbf{161} (1996)

\bibitem{wolter1997}
Wolter, F.: Fusions of modal logics revisited. In: Advances in Modal Logic. CSLI, Stanford (1997)

\end{thebibliography}


\end{document}